\pdfoutput=1 %
\RequirePackage[ARXIV]{multitarget} %
\documentclass[letterpaper,abstract=true,DIV=14,titlepage=false]{scrartcl}
\usepackage[T1]{fontenc}
\usepackage{xcolor}
\usepackage{colortbl}
\usepackage{amsopn}
\usepackage{amsmath}
\usepackage{amssymb}
\usepackage[ruled,linesnumbered,norelsize]{algorithm2e}\SetArgSty{textnormal}
\usepackage{booktabs}
\usepackage{hhline}
\usepackage{float}
\usepackage{placeins}
\usepackage[caption=false]{subfig} %
\usepackage{floatsettings}
\usepackage{setspace}
\usepackage{graphicx}
\usepackage{enumitem}
\usepackage{nicefrac}
\usepackage{anyfontsize}
\usepackage{multirow}
\usepackage{bm}
\usepackage[tt=false]{libertine}
\usepackage{mathtools}
\usepackage{makecell}
\usepackage{multirow}
\usepackage{adjustbox}
\usepackage{makecell}
\usepackage[round]{natbib}
\beginonlyworkingtarget
\usepackage{amsthm}
\endonlyworkingtarget
\beginonlypublishabletarget
\usepackage{regexpatch}
\endonlypublishabletarget
\usepackage{hyperref} %

\beginonlyworkingtarget
\newcommand{\RUNTITLE}[1]{}
\newcommand{\TITLE}[1]{\title{#1}}
\newcommand{\ARTICLEAUTHORS}[1]{\author{#1}}
\newcommand{\AUTHOR}[1]{\hspace{-0.2in}#1}
\newcommand{\AFF}[1]{\\\hspace{-0.1in}#1}
\newcommand{\EMAIL}[1]{\\\hspace{-0.1in}#1\and}
\newcommand{\ABSTRACT}[1]{\begin{abstract}#1\end{abstract}}
\newcommand{\KEYWORDS}[1]{\textbf{Keywords:} #1\\}
\newenvironment{APPENDICES}{\appendix}{}
\date{August 31, 2021}
\endonlyworkingtarget

\beginonlypublishabletarget
\OneAndAHalfSpacedXI
\TheoremsNumberedThrough
\ECRepeatTheorems
\EquationsNumberedThrough
\MANUSCRIPTNO{}

\long\def\@makefigurecaption#1#2{\@onelinecenter{
    \EGT\FigureCaptionFontStyle \HD{9}{0}{\FigureNameFontStyle #1
      \kern16pt}\ignorespaces #2\HD{0}{0}}\endgraf}
\endonlypublishabletarget

\hypersetup{
  colorlinks    = true,
  citecolor     = blue,
  linkcolor     = blue,
  urlcolor      = blue,
  pdfinfo={ Title={Machine Learning-powered Iterative Combinatorial Auctions},
            Author={Gianluca Brero, Benjamin Lubin, Sven Seuken}}
}

\DeclareMathOperator*{\eff}{Eff}

\beginonlyworkingtarget
\newtheorem{remark}{Remark}
\newtheorem{definition}{Definition}
\newtheorem{proposition}{Proposition}
\newtheorem{example}{Example}

\newtheorem{assumption}{Assumption}
\DeclareMathOperator*{\argmin}{arg\,min}
\DeclareMathOperator*{\argmax}{arg\,max}
\endonlyworkingtarget

\beginonlypublishabletarget
\let\oldproof\proof
\renewcommand{\proof}{\oldproof{Proof.}}
\let\oldendproof\endproof
\renewcommand{\endproof}{\hfill\halmos\oldendproof}
\endonlypublishabletarget

\DeclarePairedDelimiter\floor{\lfloor}{\rfloor}

\newcommand{\Reals}{\mathbb{R}}

\newenvironment{tablesizeadjustment}{}{}

\makeatletter
\newcommand{\shorteq}{%
  \settowidth{\@tempdima}{-}%
  \resizebox{\@tempdima}{\height}{=}%
}
\makeatother

\makeatletter
\newcommand{\nosemic}{\renewcommand{\@endalgocfline}{\relax}}%
\newcommand{\dosemic}{\renewcommand{\@endalgocfline}{\algocf@endline}}%
\let\oldnl\nl%
\newcommand{\nonl}{\renewcommand{\nl}{\let\nl\oldnl}}%
\makeatother

\begin{document}

\beginmaincontent
\externalrefsupplement %

\RUNTITLE{Machine Learning-powered\\Iterative Combinatorial Auctions}
\TITLE{Machine Learning-powered\\Iterative Combinatorial Auctions%
  \thanks{This paper is a significantly extended/revised version of
    \cite{brero2018MLelicit} which was published
    in the conference proceedings of IJCAI-ECAI'18.}}
\ARTICLEAUTHORS{%
\AUTHOR{Gianluca Brero\footnote{All authors have contributed equally to this work.}}
\AFF{Harvard University \EMAIL{gbrero@g.harvard.edu}}
\AUTHOR{Benjamin Lubin\footnotemark[2]}
\AFF{Boston University \EMAIL{blubin@bu.edu}}
\AUTHOR{Sven Seuken\footnotemark[2]}
\AFF{University of Zurich \EMAIL{seuken@ifi.uzh.ch}}
}

\maketitleworking

\ABSTRACT{%
\noindent %
We present a machine learning-powered iterative combinatorial auction
(MLCA). The main goal of integrating machine learning (ML) into the
auction is to improve preference elicitation, which is a major
challenge in large combinatorial auctions (CAs). In contrast to prior work, our auction design uses \emph{value queries} instead of prices to drive the auction. The ML algorithm is used to help the auction decide which value queries to
ask in every iteration. While using ML inside a CA introduces new
challenges, we demonstrate how we obtain a design that is individually
rational, satisfies no-deficit, has good incentives, and is computationally practical. We benchmark our new auction against the well-known combinatorial clock
auction (CCA). Our results indicate that, especially in large domains,
MLCA can achieve significantly higher allocative efficiency than the CCA, even with only a small number of value queries.}

\KEYWORDS{Combinatorial Auctions, Machine Learning, Preference Elicitation, CCA}

\maketitlepublishable

\vspace{-1em}

\section{Introduction}
Combinatorial auctions (CAs) are used to allocate multiple items among
multiple bidders who may view these items as complements or
substitutes. Specifically, they allow bidders to submit bids on
\emph{bundles} of items to express their complex preferences. CAs have
found widespread real-world applications, including for the sale of
spectrum licenses \citep{Cramton2013SpectrumAuctionDesign} and the
allocation of TV-ad slots \citep{Goetzendorf:2015}.

\subsection{Preference Elicitation in CAs}

One of the main challenges when conducting CAs in practice is that the
bundle space grows exponentially in the number of items, which
typically makes it impossible for bidders to fully reason about, let
alone report their full valuation, even in medium-sized domains. This
makes preference elicitation a key challenge, especially in large CAs.
Researchers have addressed this challenge by designing \emph{bidding
  languages} that are succinct yet expressive for specific
(restricted) classes of valuation functions (see
\cite{nisan2006bidding} for a survey).
But, unfortunately, if one must support general valuations and
guarantee full (or even approximate) efficiency, then the auction
requires an exponential amount of communication in the worst case
\citep{nisan2006communication}.

To get around this difficulty in practice, the preference elicitation
challenge is often addressed by using \emph{iterative combinatorial
  auctions (ICAs)}, where the auctioneer elicits information from
bidders over multiple rounds, imposing some kind of limit on the
amount of information exchanged. Common formats include
\emph{ascending-price auctions} and \emph{clock auctions}, where
prices are used to coordinate the bidding process. In recent years,
the \emph{combinatorial clock auction (CCA)} \citep{ausubel2006clock}
has gained momentum. Between 2012 and 2014 alone, ten countries have
used the CCA for conducting spectrum auctions, raising approximately
\$20 billion in revenues, with the 2014 Canadian 700 MHZ auction being
the largest auction, raising more than \$5 billion
\citep{ausubel2017practical}. Furthermore, the CCA has been used to
auction off the rights for building offshore wind farms for generating
green energy \citep{ausubel2011WindRights}.  The CCA consists of two
main phases: in the initial clock phase, the auctioneer quotes
(linear) item prices in each round, and bidders are asked to respond
to a \emph{demand query}, stating their profit-maximizing bundle at
the quoted prices. In the second phase (the supplementary round),
bidders can submit a finite number of ``all-or-nothing'' bundle bids
(typically up to 500). The goal of this design is to combine good
price discovery in the clock phase with good expressiveness in the
supplementary round.

Despite the practical successes of the CCA, a recent line of papers
has revealed some shortcomings regarding both of its phases.  Recall
that in the clock phase (which often has more than 100 rounds
\citep{canadian2014}), bidders must answer demand queries, i.e., they
must respond with their optimal bundle given a vector of
prices. However, experimental studies by
\citet{scheffel_etal_2012_OnTheImpactOfCognitiveLimits} and
\citet{bichler2013core} have shown that bidders may not be able to
optimally respond to a demand query. In particular, these studies
found that bidders tend to focus on a limited search space consisting
of some bundles of items selected prior to the auction, and that this
can cause significant efficiency losses (between 4\% to 11\% in their
experiments). Furthermore, in the supplementary round, bidders face
the difficult challenge of deciding which additional bids to submit
while obeying a limit on the maximum number of bids they can
submit.\footnote{For example, in the 2014 Canadian spectrum auction,
  bidders could submit at most 500 bids (which is a small fraction of
  the $2^{98}$ bundles), and multiple bidders reached this limit
  \citep{canadian2014}.} If the bidders do not pick their bundles well
in this round, this can further reduce efficiecy. Thus, designing a
practically feasible CA with high empirical efficiency remains a
challenging research problem.

\subsection{Identifying Useful Bundles via Machine Learning}

With general valuations, finding the efficient (or even approximately efficient) allocation in a CA requires communication that is exponential in the number of items in the worst case \citep{nisan2006communication}. This
implies that a bidder may need to answer an exponential number of
queries and/or report an exponential number of values. Given that
bidders have cognitive costs for determining
their value for a given bundle and for answering queries, no practical CA mechanism can guarantee full efficiency beyond small toy domains.
Note that the results from \cite{nisan2006communication} imply that
\textit{iterative VCG mechanisms}
\citep[e.g.,][]{mishra2007ascending,de2007ascending} also require
exponential communication in the worst case and are thus impractical
even in medium-sized domains.  %
To obtain a practical CA, we therefore design a mechanism that, like
iterative VCG, interacts with bidders over multiple rounds, but that
imposes a \emph{limit} on the information exchanged between each
bidder and the auction.

The amount and the type of information being exchanged depends on the
auction format and thus the type of queries that are used. Many
iterative CAs use demand queries, where bidders are shown prices and
have to report their profit-maximizing bundle at these prices. The
amount of information exchanged per demand query depends on the
dimensionality of the prices. For example, the CCA uses linear prices
(i.e., one price per item), which, combined with a bound on the number
of rounds, ensures practicality. Due to the result by
\cite{nisan2006communication}, this also implies that the CCA cannot
provide efficiency guarantees.

In our mechanism, the information exchanged consists of a small set of
bundle-value pairs from each bidder, and the size of this set is
limited via a parameter of our mechanism (e.g., 500). Of course, due
to \cite{nisan2006communication}, and as in the CCA, this implies that
our mechanism also cannot provide useful efficiency
\emph{guarantees}. However, our goal is to design a mechanism that
maximizes \emph{empirical efficiency} in realistic CA domains.

Intuitively, a mechanism that wants to maximize empirical efficiency
but is constrained on information exchange should elicit the ``most
important'' information from each bidder.
However, in an iterative mechanism, identifying the ``next most-useful
query'' is not an easy task because it requires reasoning about the
incremental value of the missing information given what has already
been elicited. To this end, there are three basic approaches.

First, if enough \emph{structure} regarding bidders' valuations can be
assumed ex ante, then the auctioneer can exploit this structure by
adopting an (approximately) expressive domain-specific bidding
language \citep{Sandholm2013,Goetzendorf:2015}.  One can view such
languages as value models for the domain, where the parameters of the
model are left to be specified by the bidders.  The model then ensures
that this finite amount of information is \emph{generalized}, such
that the value for all bundles can be computed.  However, in many
domains, like spectrum auctions, bidders' valuations exhibit such a
rich and complex structure that the auctioneer cannot propose a single
domain-specific language/model that would be parsimonious enough in
its own description that it could be put into practice.
Furthermore, any particular choice of a domain-specific bidding
language may favor one bidder over another, which is problematic for a
government-run auction which should not be biased against any given
participant. In the present paper, we therefore consider
\textit{general} valuations such that our mechanism does not rely on
any critical structural assumptions.

Second, instead of the auctioneer performing the generalization task
by exploiting structural assumptions, the bidders can do it themselves
by explicitly providing information statements about their whole
valuation. The CCA follows this principle by asking bidders demand
queries in each round of the clock phase. An answer to a demand query
does not only provide the auctioneer with information on one bundle,
but it is simultaneously a statement regarding the bidder's whole
valuation. The auctioneer can then use these ``global information
statements'' to identify the next query. For example, the CCA computes
the next query via a simple price update rule where prices on
over-demanded items are increased by a fixed percentage
\citep{ausubel2017practical}. However, this places a relatively large
burden on the bidders, who may not be able to respond optimally to
every demand query, as discussed above.

In this paper, we propose a third approach to perform the
generalization task, using \emph{machine learning (ML)}. The idea is
to train an ML algorithm for each individual bidder on that bidder's
value reports to \emph{generalize} to the whole bundle space, i.e., to
\emph{learn} a value for each bundle the bidder has not yet evaluated.
To build some intuition, consider the use of \emph{linear regression}
for the ML algorithm, and the following simple example: a bidder who
reports value 1 for the bundle $(1,0)$ and value $10$ for bundle
$(1,1)$. In this example, linear regression would predict a value of 9
for bundle $(0,1)$. Of course, linear regression can only learn
additive valuations as it only has one coefficient per item. In
Section~\ref{sec:MLAlgorithm}, we present our general ML framework and
explain how non-linear ML algorithms (which can capture substitute and
complements valuations) can be employed in our mechanism.
Because such an ML model provides information about the whole
valuation, this generalization can be leveraged to determine a good
sequence of bundles to query.  This is the key ingredient of the new
auction design we develop in this paper.

\subsection{Overview of Contributions}

We propose a \emph{machine learning-powered iterative combinatorial
  auction (MLCA).} In contrast to the CCA, our auction does not use
demand queries (i.e., prices) but instead asks bidders for \emph{value
  reports} on individual bundles. While in the clock phase of the CCA,
clock prices do the job of coordinating bidders towards finding an
(approximately) efficient allocation, in our approach, we use an
elicitation process guided by ML to serve this role.

Our main innvoation is an \emph{ML-powered query module} (Section
\ref{subsec:QueryModule}) which, given a set of bundle-value pairs
already reported by each bidder, determines a new query for each
bidder to be answered in the next round of the auction. The query
module consists of two key steps. First, we use an ML algorithm for
each bidder to compute a \emph{learned valuation} for this
bidder. This provides us with a prediction of each bidder's value for
any bundle in the bundle space and thus allows us to compute the
predicted social welfare of any allocation. Second, we solve an
\emph{ML-based winner determination problem} to find the allocation
with the highest predicted social welfare. The next query for each
bidder is then chosen as the bundle the bidder would be allocated in
that allocation. These two steps embody the main idea of MLCA: we aim
to generate vectors of queries that are feasible allocations, and in
particular, we query those feasible allocations that have highest
predicted social welfare.

Our second contribution is the design of the full MLCA mechanism
(Section \ref{subsec:MechanismDescription}). MLCA uses the query
module as a sub-routine in every round of the auction and ultimately
determines an allocation and payments. To achieve good incentives in
practice, MLCA aims to behave as ``similarly as possible to VCG,''
while  obeying the maximum number of bids constraint. In
addition to using the query module, the following four design elements
are essential for this: (1) the final allocation and payments are
computed based on the elicited bundle-value pairs (and not based on
learned values), (2) MLCA computes VCG payments based on these
elicited bundle-value pairs, (3) throughout the auction, MLCA
explicitly queries each bidder's marginal economy, and (4) MLCA
enables bidders to ``push'' information to the auction which they deem
useful.

In Section~\ref{sec:theoreticalAnalysis}, we provide a detailed
theoretical analysis of MLCA. First, we derive a relationship between the learning error of the ML algorithm used and the efficiency achieved by MLCA (Section~\ref{sec:learningError}). Next, we analyze the incentives of MLCA, and we explain how its design features promote good incentives in practice. Finally, we show that MLCA satisfies individual rationality and no-deficit.

In Section~\ref{sec:MLAlgorithm}, we instantiate the ML algorithm. In principle, any ML algorithm could be used within MLCA. However, a key requirement is that one can relatively quickly solve the corresponding \emph{winner determination
  problem (WDP)} based on bidders' learned valuations, to compute a
new allocation -- and this computation is done many hundreds of times
(inside the query module) throughout one run of MLCA. We
first explain how standard \emph{linear regression} can be used to
obtain the bidders' learned valuations, and how formulating the linear
regression-based WDP is straightforward. Then we generalize this to
non-linear machine learning models that can also capture substitutes
and complements valuations -- in particular, we use \emph{kernelized
  support vector regression (SVR)}. Using a different \textit{kernel} in the SVR essentially leads to a different ML algorithm inside MLCA.
We discuss SVRs with four popular kernels: Linear (which corresponds to linear regression), Quadratic, Exponential, and Gaussian, and we provide WDP formulations for all four.

In Section~\ref{sec:OptMLExp}, we then provide experimental results comparing the performance of these four ML algorithms. Our main insight in this section is that the choice of the right ML algorithm used inside MLCA
depends on two characteristics: first, the ML algorithm's prediction
performance (measured in terms of the learning error), and second, our
ability to solve the ML-based WDP problem reasonably quickly.
We find that the Quadratic kernel only has
slightly worse prediction performance than the Exponential or Gaussian kernels, but we can formulate the corresponding WDP as a Quadratic Integer program which we can solve relatively quickly, even for large auction instances. This is why, in terms of economic efficiency (our main variable of interest), the
Quadratic kernel outperforms the other three kernels.

In Section~\ref{sec:MLCAvsCCAExp}, we present experimental results regarding the efficiency of MLCA when using the Quadratic kernel (compared to the CCA) and regarding its manipulability.  For these experiments, we employ the spectrum auction test suite (SATS) \cite{weiss2017sats}, and we use
the LSVM, GSVM and MRVM domains.
Our first result is that, as one would expect, the performance of MLCA
depends on how well the ML algorithm can capture the structure of the
domain. LSVM has a highly non-linear structure, but bidders are not
interested in a large number of bundles. In this domain, the CCA
achieves slightly higher efficiency than MLCA with 500
queries.\footnote{In work subsequent to this paper,
  \citet{Weissteiner2020DeepLearning} showed how to use neural
  networks instead of SVRs as the ML algorithm in MLCA, and they
  demonstrated that this leads to an additional efficieny increase in
  the LSVM domain. Similarly, also in work subsequent to this paper,  \cite{Weissteiner2020FTMLCA} showed how Fourier analysis can be leveraged for the design of ML-powered iterative CAs to achieve better results than MLCA.} In contrast, in GSVM, the Quadratic kernel can
perfectly capture the structure of the domain, and thus is able to
generalize very well. This leads to an efficiency of 100\% for MLCA,
even with only 100 queries, matching the efficiency of the CCA.
We test both mechanisms on the MRVM domain, which is the
largest and most realistic domain, with 10 bidders and 98 items.
Here, MLCA clearly outperforms the CCA. With 500 queries, the CCA
achieves an efficiency of 94.2\% while MLCA achieves 96.4\%, which is
a 2.2\% point improvement. Finally, we provide experimental evidence for MLCA's robustness against strategic manipulations in Section~\ref{subsec:ManipulationExperiments}. Overall, our results suggest that MLCA may be a promising candidate for running large-scale CAs in practice.

\subsection{Related Work}

An important research agenda related to this paper is the one on
preference elicitation in combinatorial auctions, which dates back to
the early 2000s \citep{sandholm2006preference}. Among the most
relevant work in this agenda are the papers by \citet{lahaie2004applying} and \citet{blum2004preference} which used learning algorithms to design elicitation algorithms for specific classes of valuations. Elicitation algorithms for generic valuations based on ML were introduced by \citet{lahaie2011kernel} and further developed by \citet{lahaie2019adaptive}. These approaches are based on
using ML on bidders' reports to find \emph{competitive equilibrium
  prices}, i.e., prices that coordinate bidders towards demanding
efficient allocations. \citet{brero2018Bayes} and
\citet{brero2019Fast} built on these mechanisms and designed Bayesian
auctions that, besides bidders' reports, can exploit prior beliefs on
bidders' valuations to quickly determine competitive equilibrium
prices. In contrast to the pre-dominant research agenda in this field,
we do not use ML to directly learn or predict competitive equilibrium
prices. Instead, we use ML to learn bidders' valuations, and we only
provide an implicit correspondence to approximate clearing
price.

Another related line of research grew out of the automated mechanism design (AMD) research agenda \citep{Conitzer2002AMD,conitzer2004self}, aiming to use algorithms to design \textit{direct-revelation} mechanisms: these are mechanisms
where agents report all of their preferences upfront and the mechanism
determines the outcome based on these preferences.  The first
approaches to AMD were based on formulating the mechanism design
problem as a search problem over the space of all possible mappings
from agents' preference profiles to outcomes. These approaches were
only applicable to small settings mainly because of the dimension of
the preference profile space. Recent work has partly addressed this
scalability issue by limiting the search to parametrized classes of
mechanisms and using learning algorithms to find suitable parameter
values. \citet{dutting2015payment} used discriminant-based classifiers
to learn approximately strategyproof payment rules for combinatorial
auctions. In more recent work, \citet{DuettingOptimalAuctionsDL} have
used deep learning methods to advance the design of revenue-maximizing
auctions. Thus, this strand of research has used ML in the process of designing new mechanisms. In contrast, in our work, we incorporate an ML algorithm into the mechanism itself, such that an ML algorithm is part of the resulting mechanism's execution.

\section{Preliminaries}

\subsection{ Formal Model}\label{ssec:CASetting}
In a combinatorial auction (CA), there is a set of a set of $m$
indivisible items denoted $M=\{1,...,m\}$ to be allocated among $n$
bidders denoted $N=\{1,...,n\}$.  A bundle is a subset of the set of
items. We associate each bundle with its indicator vector and denote
the set of bundles as $\mathcal X = \{0,1\}^m$. We represent the
preferences of each bidder $i$ with a value function $v_i: \mathcal{X}
\rightarrow \mathbb{R}_{\geq 0} $ that is private knowledge of the
bidder. Thus, for each bundle $x \in \mathcal{X}, v_i(x)$ represents
the \emph{true value} that bidder $i$ has for obtaining $x$. We do not
make any assumptions regarding the structure of bidders' value
functions.\footnote{However, our mechanism allows for structural
  assumptions (like free disposal) if they are needed.} We assume that
values are normalized such that the bidders have zero value for the
empty bundle. We also refer to $v_i$ as bidder $i$'s \emph{valuation}.
We let $v = (v_1,\dots,v_n)$ denote the \emph{valuation profile} and
$v_{-i} = (v_1,\dots,v_{i-1},v_{i+1},\dots,v_n)$ the corresponding
profile where bidder $i$ is excluded.  We also call the set of all
bidders $N$ the \emph{main economy}, and we call the set $N \setminus
\{i\}$ bidder $i$'s \emph{marginal economy.}

A CA \emph{mechanism} defines how bidders interact with the auction,
how the final allocation is determined, and how much each bidder has
to pay.  By $a = (a_1, \dots , a_n) \in \mathcal{X}^n$ we denote an
\emph{allocation}, with $a_i$ being the bundle that bidder $i$ obtains
under $a$. We denote the set of feasible allocations by $\mathcal
F$. Our mechanism supports various notions of feasibility; in this
paper, we consider an allocation to be \textit{feasible} if each item
is allocated to at most one bidder. \emph{Payments} are defined as a
vector $p = (p_1,...,p_n)\in \Reals^n$, with $p_i$ denoting the amount
charged to bidder $i$. We assume that bidders have a
\emph{quasi-linear utility function} of the form 
\begin{equation}\label{eq:quasiLinearUtility}
    u_i(a,p) = v_i(a_{i})-p_i.    
\end{equation}
The \emph{social welfare} of allocation $a$ is
defined as the sum of the bidders' true values for $a$, i.e., $V(a) =
\sum_{i \in N} v_i(a)$. The social welfare-maximizing (i.e.,
\emph{efficient}) allocation is denoted as $a^* \in \argmax_{ a \in
  \mathcal F} V(a)$. We measure the efficiency of any allocation $a\in
\mathcal F$ as $\eff(a)=V(a)/V(a^*)$. We aim to design mechanisms that
allocate items such that this measure of efficiency is maximized.

 In this paper, we design a mechanism for an \emph{iterative
   combinatorial auction} that asks each bidder to make value reports
 for different bundles across different rounds of the auction.  We let
 $(x_{ik}, \hat{v}_{ik})$ denote the $k^{\textrm{th}}$
 \emph{bundle-value report} from bidder $i$, where $x_{ik}$ denotes
 the bundle and $\hat v_{ik}$ denotes the corresponding value report
 (which may not necessarily be truthful). Throughout the auction, our
 mechanism maintains the \emph{set of bundle-value reports} (or
 \emph{report set}, for short) from bidder $i$ denoted $R_i$. At the
 point when bidder $i$ has made $\ell$ bundle-value reports, we have
 $R_i=\{(x_{ik},\hat{v}_{ik})\}_{k \in L}$ where
 $L=\{1,...,\ell\}$. Note that this notation enables us to refer to
 specific bundle-value pairs contained in $R_i$ via an index. We let
 $R = (R _1,\dots,R_n)$ denote the profile of these report sets. For
 notational simplicity, we say that a bundle $x\in R_i$ if bidder $i$
 has made a value report on bundle $x$.
 We use $|R_i|$ to denote the
 number of bundle-value reports made by bidder $i$. Our mechanism will
 enforce that a bidder cannot make multiple reports on the same
 bundle.

 At the end of the auction (after the last round), our mechanism
 determines a final allocation and payments based on all elicited
 value reports.  In this step, we need to ``look up'' a bidder's value
 report for specific bundles. To this end, we define the \emph{report
   function} $\hat{v}_i(\cdot)$, such that, for any bundle $x \in
 R_i$, $\hat{v}_i(x)$ is bidder $i$'s value report for $x$ (and
 undefined otherwise). Note that this definition of $\hat{v}_i(\cdot)$
 is different from its usage in the literature on direct revelation
 mechanisms where this function (when applied to CAs) assigns a value
 report to \emph{every} bundle in the bundle space.

In our mechanism, bidders typically only make a value report on a
limited number of bundles in the bundle space. Given that
$\hat{v}_i(\cdot)$ captures all of the reported information by bidder
$i$, and without making further assumptions, we only need to consider
the set of bundles restricted to the $R_i$'s when determining the
efficient allocation with respect to the $\hat{v}_i(\cdot)$'s.
Formally, we define this restricted set of feasible allocations as
\mbox{$\mathcal{F}_R=\{a\in \mathcal{F}: a_i \in R_i \;\forall i \}$}.
Consequently, given a profile of report sets $R$, our mechanism
computes the efficient allocation with respect to $R$ by solving
$\argmax_{a \in \mathcal{F}_R} \sum_{i \in N} \hat{v}_i (a_i),$ which
is a standard (combinatorial) winner determination problem that can be
formulated as an Integer Program.

\subsection{The Vickrey-Clarke-Groves (VCG) Auction}\label{ssec:CAEfficiency}

We now define the well-known VCG auction
\citep{vickrey1961counterspeculation,clarke1971multipart,groves1973incentives}.

\begin{definition}[VCG Auction]\label{def:VCG mechanism}
Under the VCG auction, every bidder $i$ must make a value report for
every bundle in the bundle space, which is captured by the report
function $\hat{v}_i(\cdot)$. Given these value reports, the outcome is
defined by:\begin{itemize}
\item   The allocation rule:  $a^{\text{VCG}} \in \argmax_{a \in \mathcal{F}} \sum_{i \in N} \hat{v}_i(a_i)$
\item The payment rule, which charges every bidder $i$ a payment:
\begin{equation} \label{eq:payment_VCG}
p^{\text{VCG}}_i = \sum_{j \in N\setminus \{i\}} \hat v_j (a^{-i}_j) \! - \sum_{j \in N\setminus \{i\}} \hat v_j ( a^{\text{VCG}}_j) , \text{\hspace{0.5cm} where } a^{-i}\in \argmax_{a \in \mathcal{F}} \sum_{j \in N \setminus \{i\}} \hat{v}_j (a_j)  \, .
\end{equation}
\end{itemize}
\end{definition}

It is well known that the VCG auction is \emph{strategyproof}, i.e., it is a
dominant strategy for all bidders to report their true valuation,
independent of what the other bidders do. 
When all bidders follow their (truthful) dominant strategy, the VCG auction is efficient.

Note that the VCG auction is often impractical for CAs, as it
requires bidders to report their full valuation (which is
exponentially sized in the number of items). 
Thus, researchers have proposed \emph{iterative VCG mechanisms} \citep[e.g.,][]{mishra2007ascending,de2007ascending} that interact with bidders over multiple
rounds and only elicit enough information to determine the VCG
outcome. By realizing the VCG outcome, these mechanisms inherit some good incentive properties of the VCG auction, supporting truthful bidding in ex-post Nash equilibrium. However, iterative VCG mechanisms are still impractical in general, because, in the worst case, they need to exchange an exponentially sized amount of information with bidders to determine an efficient allocation \citep{mishra2007ascending}.

\section{The  Machine Learning-powered Combinatorial Auction (MLCA)}
\label{sec:MLCA}

In this section, we introduce our new \emph{ML-powered combinatorial
  auction (MLCA)}.  To obtain a practical CA, we design a mechanism
that, like iterative VCG mechanisms, interacts with bidders over
multiple rounds. However, in contrast to iterative VCG mechanisms, we
impose an important \emph{design constraint} on the amount of
information exchanged between each bidder and the auction: the
auctioneer limits the maximum number of bids which each bidder is
allowed to submit, and this maximum is a parameter of our design
(e.g., 500). Under this constraint, we know from
\cite{nisan2006communication} that we cannot provide useful efficiency
guarantees, and we cannot directly apply iterative VCG mechanisms to
obtain ex-post Nash incentive compatibility. Therefore, we adopt as
our goals for MLCA to (a) maximize \emph{empirical efficiency} in
realistic CA domains and (b) have good incentives in practice.

To achieve high empirical efficiency despite the limit on the number
of bids, our most important innovation is the \emph{ML-powered query
  module}, which we introduce in Section~\ref{subsec:QueryModule}. In
Section \ref{subsec:MechanismDescription} we then present the full
MLCA mechanism, which uses the query module as a sub-routine in every
round of the auction, and which ultimately determines the final
allocation and payments. To achieve good incentives in practice, we
combine the query module with four additional design elements to
obtain a mechanism that (in practice) behaves as ``similarly as
possible to VCG,'' while always obeying the maximum number of bids
constraint. Concretely, the four design elements are: (1) the final
allocation and payments are computed based on the elicited
bundle-value pairs (and not based on learned values), (2) MLCA
computes VCG payments based on these elicited bundle-value pairs, (3)
throughout the auction, MLCA explicitly queries each bidder's marginal
economy, and (4) MLCA enables bidders to ``push'' information to the
auction which they deem useful. We analyze the effect these design
elements have on incentives in Section~\ref{SEC:Incentives} and we provide experimental evidence for MLCA's robustness against strategic manipulations in Section~\ref{subsec:ManipulationExperiments}.

\subsection{Machine Learning-powered Query Module}
\label{subsec:QueryModule}

\begin{figure}[tb!]
        \centering
        \includegraphics[width=0.6\linewidth]{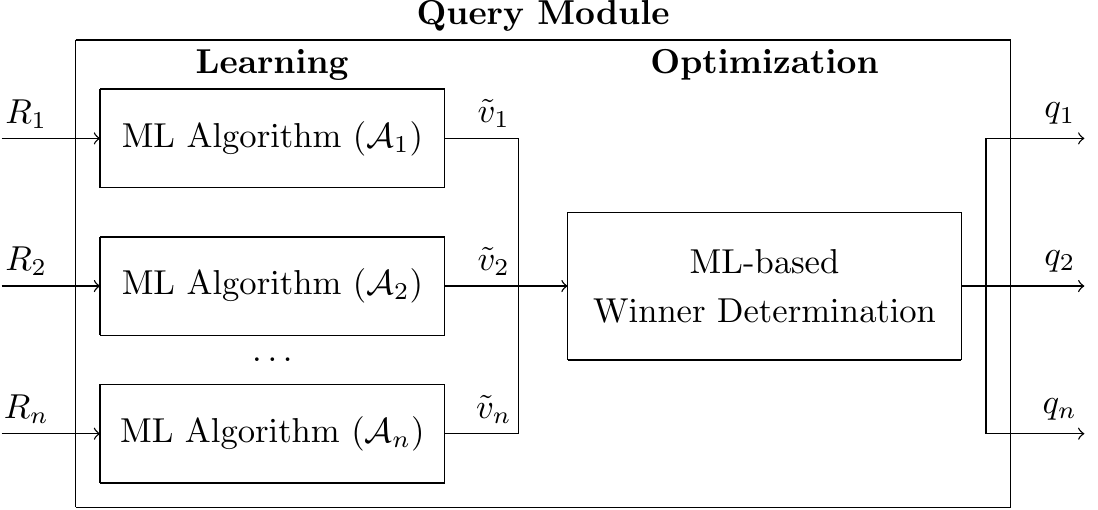}
        \caption{Schematic representation of how the query module works.}
        \label{fig:elicitation}
\end{figure}

  In this section, we present the \textit{ML-powered
  query module}, which is based on three key ideas. First, we use ML
to learn each bidder's full value function from a small set of
reported bundle-value pairs, such that we can predict the bidder's
value for any bundle in the bundle space. Second, we aim to generate
query profiles that are feasible allocations. Third, putting these
ideas together, we design an ML-based optimization algorithm which,
based on the bidders' learned valuations, finds feasible allocations
with high predicted social welfare.

We provide a high-level overview of the query module in
Figure~\ref{fig:elicitation}. The module takes as input a set of
bundle-value pairs $R_i$ from each bidder $i$. In a first step, for
each bidder $i$, we train a separate ML algorithm $\mathcal{A}_i$
based on $R_i$ to learn $i$'s valuation $\tilde{v}_i$ (which we call
the \emph{learned valuation} going forward).  Given all bidders'
learned valuations, we define the \emph{learned social welfare} of any
allocation $a$ as follows:
\begin{align}
\tilde{V}(a) = \sum_{i\in N} \tilde{v}_i(a_{i})
\end{align}
In a second step, we then feed all learned valuations $\tilde{v}_i$ to
an \emph{ML-based winner determination algorithm}. This algorithm
finds a candidate allocation $\tilde{a}$ that maximizes the learned
social welfare $\tilde{V}(\tilde{a})$. Finally, the query module
outputs a query $q_i$ for each bidder equal to the bundle he would be
allocated under this candidate allocation, i.e., $q_i = \tilde{a}_i$.

To understand the design of our query module, recall that, in our
mechanism, we use ML only to guide the elicitation of the most useful
bundles, but the computation of the final allocation and payments is
only based on the \textit{reported} bundle-value pairs.  Thus, by
generating a query profile $q=\tilde{a}$ that constitutes a feasible
allocation, we are gathering bids on bundles that could in the end be
used in the final allocation.  By contrast, if we generated each
bidder's query by only looking at this individual bidder, we would
likely generate queries that are ``incompatible'' (i.e., the bundles
are overlapping) and thus could not be used in the final allocation.

Algorithm~\ref{alg:ML_Elicitation} provides a formal description of
our ML-powered query module -- to improve clarity and readability, we
here present a slightly abbreviated description; we provide a full
version of the query module in Algorithm~\ref{alg:ML_Elicitation_Full}
in the \refapp{app:FullVersionsOfAlgorithms}. Formally, we
represent the query module as a function
\emph{NextQueries}$_\mathcal{A}(I,R)$. The function is parameterized
by a profile of machine learning algorithms $\mathcal{A}$, one for
each bidder.\footnote{In Section~\ref{sec:linearRegression}, we first
  instantiate the ML algorithm using linear regression; in
  Section~\ref{sec:SVR} we then generalize to non-linear learning
  models by using kernelized support vector regression. For our
  experiments (Sections~\ref{sec:OptMLExp} and \ref{sec:MLCAvsCCAExp})
  we equip all bidders with the same ML algorithm. However, our
  approach also supports using a different ML algorithm for each
  bidder (see Remark~\ref{rem:DifferentMLForDifferentBidders} in
  Section~\ref{sec:SVR}).} The function takes as its first argument an
index set $I \subseteq N$ of the bidders for the economy to be
considered (as we will describe in
Section~\ref{subsec:MechanismDescription}, our mechanism calls the
query module for both the main and all bidders' marginal
economies). The function takes as its second argument a profile of
reports $R$, where each $R_i$ is the set of bundle-value pairs that
bidder $i$ has already reported.

\begin{algorithm}[tb]
        \SetEndCharOfAlgoLine{}
        \SetKwRepeat{Do}{do}{while}
        \nonl\textbf{function} NextQueries$_\mathcal A$$(I,R)$;\;
        \nonl \textbf{parameters:} {profile of ML algorithms $\mathcal{A} $};\;
        \nonl \textbf{inputs:} {index set of bidders for the economy to be considered $I$; profile of reports $R$};\;
        \textbf{foreach} bidder $i\in I$ : $\tilde v_{i} := \mathcal{A}_i (R_{i})$;   \text{\hspace{0.5cm}$\setminus\setminus $\textbf{Learning Step:} learn valuations using ML algorithm} \label{AlgQueryModule:LearningStep}\;
        select $\tilde{a} \in \argmax_{a \in {\mathcal F}} \sum_{i\in I}\tilde v_{i} (a_{i})$; \text{\hspace{0.5cm}$\setminus\setminus $\textbf{Optimization Step} (based on learned valuations)} \label{AlgQueryModule:OptStep}\;
        assign new query profile: $q=\tilde{a}$ (i.e., for each $i \in I: q_i = \tilde{a}_i$); \label{AlgQueryModule:FirstQueryAssignment}\;%
        \ForEach
        {$i \in I$}
        {
                \uIf  {bundle $q_i$ has already been queried from bidder $i$ (i.e., $q_i \in R_i)$ \label{AlgQueryModule:QueryCheck}}{
                        define set of allocations containing a new query for bidder $i$: $\mathcal F' :=  \{a\in \mathcal F : a_i \neq x, \forall x\in R_i\, \}$; \label{AlgQueryModule:NewQuerySet}\;
                        select  $\tilde{a} \in \argmax_{a \in {\mathcal F'}} \sum_{i'\in I}\tilde v_{i'} (a_{i'})$; \text{\hspace{0.5cm}$\setminus\setminus $\textbf{Optimization Step} (with restrictions) } \label{AlgQueryModule:OptStepRestrictions}\;
                        overwrite new query for bidder $i$: $q_i = \tilde{a}_i$; \label{AlgQueryModule:OverwriteQuery}\;
                }
        }
        \textbf{return} profile of new queries $q$; \label{AlgQueryModule:ReturnQuery}\;
        \caption{Machine Learning-powered Query Module}
        \label{alg:ML_Elicitation}
\end{algorithm}

In Line~\ref{AlgQueryModule:LearningStep}, we begin with the
\textit{learning step}: for each bidder $i \in I$, bidder $i$'s ML
algorithm $\mathcal A_i$ is trained on the bundle-value pairs $R_i$ to
learn a valuation $\tilde{v}_i$.  In
Line~\ref{AlgQueryModule:OptStep}, in the \textit{optimization step},
we solve the ML-based winner determination problem, i.e, given all
bidders' \textit{learned valuations} $\tilde{v}_i$, we find an
allocation $\tilde{a}$ that maximizes the learned social welfare
$\tilde{V}(\tilde{a})$. In Section~\ref{sec:MLAlgorithm}, we show how
this ML-based winner determination problem can be solved in practice.
In Line~\ref{AlgQueryModule:FirstQueryAssignment}, we assign $q$
(i.e., the candidate query profile) to the allocation $\tilde{a}$
found in the optimization step. Note that at this point, all bidders'
queries come from the same allocation $\tilde{a}$ and are thus
compatible. However, some bidders may already have reported the value
for their candidate query $q_i = \tilde{a}_i$. Because we want to
generate a \textit{new} query for every bidder, we next test for this
on Line~\ref{AlgQueryModule:QueryCheck}, i.e., whether $q_i \in
R_i$. If true, we then create a ``next-best query'' for that bidder:
in Line~\ref{AlgQueryModule:NewQuerySet}, we first define a set of
feasible allocations $\mathcal{F}'$ specifically tailored to bidder
$i$, making sure that, for each allocation $a$ in the set, bidder $i$
has not yet reported his value for bundle $a_i$.  In
Line~\ref{AlgQueryModule:OptStepRestrictions}, we then solve a new
\textit{restricted} optimization problem, this time solving the
ML-based winner determination problem, but now with the feasible
allocations restricted to the set $\mathcal{F}'$ defined in the
previous step. In Line~\ref{AlgQueryModule:OverwriteQuery}, we then
overwrite bidder $i$'s query with the bundle he is assigned in the
allocation found in the previous step.  Finally, in
Line~\ref{AlgQueryModule:ReturnQuery}, we return the final query
profile.

\begin{remark}
In MLCA, bidders submit bundle-value reports, while prices are
used for elicitation in most prior work on iterative CAs
\citep[e.g.][]{parkes2006mit}.  Typically, the goal of such
price-based elicitation is to obtain approximate clearing
prices. It is possible to relate our mechanism to the rest of the
literature through a price-based interpretation of the
elicitation performed by MLCA's query module.  Specifically, in
\refapp{app:CE} we describe how to impute approximate clearing
prices that are implicit in the elicitation.  We provide a bound
on how close these prices are to clearing prices in the learning 
error of the ML algorithm.  
We emphasize that such imputed prices are only implicit when MLCA is
run. The imputed prices will in general be
non-anonymous, and if we use high-dimensional (non-linear) ML
algorithms, then they will also be high-dimension bundle
prices. This implicit high-dimensional price structure
enables MLCA to find an approximate CE (where the approximation
depends on the learning error of the ML algorithm) while
approaches based on linear prices (such as the CCA) may be limited.
For more details on this price-based interpretation of MLCA, please
see the discussion in \refapp{app:CE}.
\end{remark}

\subsection{The MLCA Mechanism}
\label{subsec:MechanismDescription}

In this subsection, we describe our full \emph{machine
  learning-powered combinatorial auction} (MLCA) mechanism. MLCA is an
iterative auction that proceeds in rounds until a maximum number of
queries has been asked. In each round, it generates new queries for
the main economy as well as each bidder's marginal economy, each time
employing the query module described in
Section~\ref{subsec:QueryModule}. After each round, the newly
generated queries are sent to the bidders and the mechanism waits for
the corresponding bundle-value reports. This enables the query module
(when called in the next round) to update each bidder's ML algorithm
based on the new reports and exclude already reported bundle-value
reports when generating a new query. At the end of the auction, the
mechanism computes the final allocation and payments based only on all
bidders' \textit{reported} bundle-value pairs. As the default payment
rule, we compute VCG payments based on these reports.

\begin{algorithm}[tb]
        \SetEndCharOfAlgoLine{}
        \nonl \textbf{parameters:} profile of ML algorithms $\mathcal A$; maximum \# of queries per bidder $Q^{\text{max}}$;
        \\\nonl \# of initial queries $Q^{\text{init}} \leq Q^{\text{max}}$\label{alg:MLCA:Parameters};\;
        \textbf{foreach} bidder $i\in N$: ask the bidder to report his value for $Q^{\text{init}}$ randomly chosen bundles\label{alg:MLCA:InitialReports};\;
        Let $R= (R_1,...,R_n)$ denote the initial report profile; each $R_i$ is $i$'s set of bundle-value reports\label{alg:MLCAInitialReportBookkeping};\;
        Let $T=\floor{(Q^{\text{max}} - Q^{\text{init}})/n}$ denote the total number of auction rounds and $t=1$ the current round\label{alg:MLCA:TotRounds};\;
        \While(\tcp*[f]{ Auction round iterator}){$t \leq T$ \label{alg:MLCA:RoundIterator}} {
                Generate query profile via \emph{NextQueries}$_\mathcal A$$(N,R)$;\tcp*{Queries for the main economy} \label{alg:MLCA:QueryMain}
                \ForEach
                {$i \in N$ \label{alg:MLCA:BiddersIt}}
                {
                        Generate query profile via \emph{NextQueries}$_{\mathcal{A}_{-i}}$$(N\setminus\{i\},R_{-i})$;\tcp*{Queries for marg. econ} \label{alg:MLCA:QueryMarginal}
                } \label{alg:MLCA:EndFor}
                \textbf{foreach} bidder $i \in N$: send new queries to bidder $i$ and wait for reports\label{alg:MLCA:SendQueries};\;
                \textbf{foreach} bidder $i \in N$: receive bundle-value reports $R_i'$ and add them to $R_i$, i.e., $R_i = R_i\cup R_i'$\label{alg:MLCA:NewReports};\;
                $t = t+1$\label{alg:MLCA:IncreaseRound};\;
        }\label{alg:MLCA:endWhile}
        Let $\hat{v}_i(\cdot)$ denote bidder $i$'s \emph{report function} capturing his bundle-value reports $R_i$, $\forall i \in N$\label{alg:MLCA:ReportFunction};\;
        Compute final allocation: $a^{\text{MLCA}}\in \argmax_{a \in \mathcal{F}_R} \sum_{i \in N} \hat{v}_i (a_i)$\label{alg:MLCA:FinalAllocation};\;
        \textbf{foreach} bidder $i \in N$: compute his payment
        \begin{equation}
        \label{eq:payment_MLCA}
        p^{\text{MLCA}}_i = \sum_{j \in N \setminus \{i\}} \hat v_j ( a^{-i}_j) - \sum_{j \in N \setminus \{i\}}\hat
        v_j ( a^{\text{MLCA}}_j), \text{\hspace{1cm} where } a^{-i}\in \argmax_{a \in \mathcal{F}_R} \sum_{j \in N \setminus \{i\}} \hat{v}_j (a_j) ;
        \end{equation}\label{alg:MLCA:PaymentComputation}\\
        Output allocation $a^{\text{MLCA}}$ and payments $p^{\text{MLCA}}$\label{alg:MLCA:Output};\;
        \caption{Machine Learning-powered Combinatorial Auction (MLCA)}
        \label{alg:MLCA}
\end{algorithm}

Algorithm~\ref{alg:MLCA} provides a formal description of the MLCA
mechanism -- again, to improve clarity and readability, we here
present a slightly abbreviated version; we present the full version of
MLCA in Algorithm~\ref{alg:MLCA_Full} of
\refapp{app:FullVersionsOfAlgorithms}.  The mechanism is
parameterized by a profile of machine learning algorithms $\mathcal A$
(one for each bidder) that will be used to parameterize the query
module. The mechanism has two more parameters: an overall
\textit{maximum number of queries} $Q^{\text{max}}$ that it can ask
each bidder, and the \textit{number of initial queries}
$Q^{\text{init}}$ that it asks each bidder in the
\textit{initialization phase} of the auction (which we can think of as
round $0$).

When the auction begins, the mechanism asks each bidder $i$ to report
their values for $Q^{init}$ randomly chosen bundles
(Line~\ref{alg:MLCA:InitialReports}). Next, each bidder $i$ reports a
value for each bundle he is queried, and all resulting bundle-value
pairs are then stored in the variable $R_i$
(Line~\ref{alg:MLCAInitialReportBookkeping}). We include this
initialization phase
(Lines~\ref{alg:MLCA:InitialReports}-\ref{alg:MLCAInitialReportBookkeping})
before going into the iterative phase, such that the ML algorithms
used in the query module already have some amount of training data
when they are first called. Specifically, this avoids wasteful
learning and optimization steps in the query module (when the ML
algorithms have no or little training data) and additionally leads to
some amount of exploration of the bundle space.\footnote{In our
  experiments in Section~\ref{sec:MLCAvsCCAExp}, we optimize the
  number of initial queries $Q^{\text{init}}$ to maximize
  efficiency. There are many avenues for potential improvement of the
  initialization phase, for example by using a Bayesian approach to
  determine the optimal sequence of initial queries. However,
  exploring those ideas is beyond the scope of this paper and we leave
  them to future work.}

Next, we enter the \emph{iterative phase} of the mechanism. We first
define the total number of auction rounds~$T$
(Line~\ref{alg:MLCA:TotRounds}), which is simply the number of queries
left (i.e., $Q^{max} - Q^{init}$) divided by the number of agents in
the auction~$n$, because in each round, each bidder will be asked $n$
queries (one query in the main economy and $n-1$ marginal economy
queries). Next,
Lines~\ref{alg:MLCA:RoundIterator}-\ref{alg:MLCA:endWhile} iterate
over the rounds of the auction. In Line~\ref{alg:MLCA:QueryMain} we
call the function \emph{NextQueries} to generate a query profile for
the main economy (one new query for each bidder). Therefore, the
function \emph{NextQueries} takes as inputs the set of all bidders $N$
and all bidders' bundle-value reports received so far
$R$. Additionally, we parameterize \emph{NextQueries} using the
profile of ML algorithms $\mathcal{A}$. Next, in
Lines~\ref{alg:MLCA:BiddersIt}-~\ref{alg:MLCA:EndFor}, we then create
a query profile for each bidder $i$'s marginal economy. Accordingly,
in Line~\ref{alg:MLCA:QueryMarginal}, we again call the function
\emph{NextQueries}, but this time using as inputs the set of all
bidders excluding bidder $i$, and all bundle-value reports except for
those by bidder $i$. Additionally, we now parameterize
\emph{NextQueries} using
$\mathcal{A}_{-i}$.\footnote{\label{FN:BookkeepingQueries}Note that in
  our implementation of MLCA, when calling \emph{NextQueries}, we also
  hand over the set of queries that have already been generated in the
  current auction round, such that the query module can automatically
  exclude these queries in the optimization steps. This is needed to
  guarantee that the query module generates a \emph{new} query for
  every bidder. For details, please see
  Algorithms~\ref{alg:ML_Elicitation_Full} and \ref{alg:MLCA_Full} in
  \refapp{app:FullVersionsOfAlgorithms}.}

Next, we send all queries that have been generated in this round to
the bidders and wait for their reports
(Line~\ref{alg:MLCA:SendQueries}). Upon receiving each bidder's $i$'s
reports for the queried bundles, we then store the new bundle-value
pairs in the variables $R_i$
(Line~\ref{alg:MLCA:NewReports}). Finally, at the end of this auction
round, we increase the round counter by one
(Line~\ref{alg:MLCA:IncreaseRound}).

After the last round of the auction (Line~\ref{alg:MLCA:endWhile}), we
proceed towards computing the final allocation and payments. Because
these steps require us to ``look up'' a bidder's value report for
specific bundles, for notational simplicity, we first construct each
bidder $i$'s \textit{report function} $\hat{v}_i(\cdot)$ to capture
all bundle-value reports made by bidder $i$
(Line~\ref{alg:MLCA:ReportFunction}). Recall from
Section~\ref{ssec:CASetting} that $\hat{v}_i(\cdot)$ is defined in
such a way that, for any bundle $x \in R_i$, $\hat{v}_i(x)$ is bidder
$i$'s value report for $x$ (and undefined otherwise). In
Line~\ref{alg:MLCA:FinalAllocation}, we compute the final allocation
$a^{\text{MLCA}}$ by solving $\argmax_{a \in \mathcal{F}_R} \sum_{i
  \in N} \hat{v}_i (a_i)$, which is a standard (combinatorial) winner
determination problem to find the optimal allocation according to the
$\hat{v}_i$'s. Note, that the set of feasible allocations
\mbox{$\mathcal{F}_R=\{a\in \mathcal{F}: a_i \in R_i \;\forall i \}$}
is restricted to only contain bundles $a_i$ for which a bidder has
made an explicit value report.\footnote{This implies that
  $\hat{v}_i(\cdot)$ will only ever be called for bundles for which it
  is defined.} In Line~\ref{alg:MLCA:PaymentComputation} we then
compute each bidder $i$'s payment $p_i^{\text{MLCA}}$. Concretely, in
Equation~\eqref{eq:payment_MLCA}, we essentially compute VCG payments;
but in contrast to standard VCG, when computing the optimal
allocations $a^{\text{MLCA}}$ and $a^{-i}$, we restrict the set of
feasible allocations to only contain bundles for which the bidders
have explicitly reported a value.  Finally, we output the final
allocation $a^{\text{MLCA}}$ and the payment profile
$p^{\text{MLCA}}$.

\begin{remark}[Randomization and Information Hiding in MLCA]
\label{rem:InformationSharingInMLCA}
Note that there are two steps in MLCA where we use
randomization. First, in Line~\ref{alg:MLCA:InitialReports} of the
algorithm, we generate a separate randomly chosen set of queries for
every bidder, such that the bidders do not know each others'
queries. Second, in Line~\ref{alg:MLCA:SendQueries} of the algorithm,
we send the queries to the bidders as a \emph{randomized} list, such
that a bidder cannot easily identify whether a query was generated in
the main economy or in a specific bidder's marginal economy.  Of
course, we never tell any bidder about the queries nor value reports
of other bidders. These design features make MLCA less transparent to
the bidders, which makes it more difficult for a bidder to predict how
the mechanism will use a given value report. Intuitively, this makes
the mechanism more robust to manipulations, as we will discuss in
Section~\ref{SEC:Incentives}. This argument is similar to the one
presented in~\citet{parkes2001iterative}.
\end{remark}

\begin{remark}[Answering Value Queries with Upper and Lower Bounds]
Note that MLCA uses standard value queries that ask bidders to reply
with a single value report to each query. However, in practice, rather
than replying with a single (exact) value report, it may be easier for
bidders to specify \textit{bounds on bundle values} \citep[see,
  e.g.,][]{parkes2006mit}. In subsequent work, \cite{Beyeler2020iMLCA} have
recently introduced a modified version of MLCA that takes as input
bidders' reports consisting of upper and lower bounds on their values,
and they have shown that this modified version of MLCA still achieves
comparable levels of efficiency.
\end{remark}

\subsection{Additional Design Features of MLCA}

We now describe three additional design features of MLCA that will
likely be important in many domains when applying MLCA: (a) enabling
bidders to ``push bids,'' (b) enabling the auctioneer to control the
number of rounds of the auction, and (c) enabling the auctioneer to
switch out the payment rule (e.g., to charge core-selecting
payments).\medskip

\noindent\textbf{Push bids.} Recall that MLCA aims to iteratively
elicit the most useful information from the bidders by sending them
carefully chosen queries.  However, a bidder might have some
particularly useful information about his valuation that  could improve efficiency. Thus, a natural idea is to let bidders
provide unsolicited information to the mechanism. Following
\citet{sandholm2006preference}, we call this ``bidder push.''
Concretely, in MLCA, we allow bidders to submit a limited number of
\emph{push bids}, which are unsolicited bundle-value pairs that the
bidders can submit in the initialization phase.
These push bids are then treated the same as the queried bids, i.e, they are included in the learning step;\footnote{Given the different nature of these bids, it may be useful to weight these bids differently in the learning step, to avoid potential negative effects on efficiency.} and of course, they are included in the final winner determination (the
details are described in Algorithm~\ref{alg:MLCA_Full} in
\refapp{app:FullVersionsOfAlgorithms}).  Note that a second
motivation for adding push bids (beyond potentially increasing
efficiency) is related to incentives.  On this point,
\citet{nisan2007computationally} have shown that letting bidders
provide unsolicited information in a VCG-style mechanism can improve
the incentives of a mechanism. In Section~\ref{SEC:Incentives}, where
we analyze the incentives of MLCA, we will make a similar point,
arguing that push bids improve the incentives of MLCA.  Finally, in
some domains, it may be useful to let bidders push some information
about their preferences in a language that is richer than using
bundle-value pairs. For example, bidders could state ``my valuation is
additive,'' ``I am only interested in a specific subset of the
items,'' or ``item A creates strong synergies with all other items.''
MLCA can be extended to also handle such structural
preference descriptions.\medskip

\noindent\textbf{Controlling the number of rounds.} Note that the total number of
auction rounds~$T$ in MLCA depends on the maximum number of queries $Q^{\text{max}}$, the number of initial queries $Q^{\text{init}}$, and the number of agents $n$ in the following way:
$T=\floor{(Q^{\text{max}} - Q^{\text{init}})/n}$ (Line~\ref{alg:MLCA:TotRounds} of Algorithm~\ref{alg:MLCA}). In practice, the auctioneer cannot control any of these parameters without hindering the auction performance: $n$ is determined by the setting at hand, $Q^{\text{max}}$ is maximized under the constraint of maintaining a reasonable bidding effort, and $Q^{\text{init}}$ is set to be the optimal amount of exploratory queries made during the initialization phase of the auction. To provide the auctioneer with better control over the number of auction rounds, we deploy MLCA with an additional parameter $Q^{\text{round}}$ that controls the number of queries MLCA asks each bidder at each round: When $Q^{\text{round}}$ is larger than the number of bidders $n$, the auction would go through the steps of generating queries for the main and marginal economies (Lines~\ref{alg:MLCA:QueryMain}--\ref{alg:MLCA:EndFor})
multiple times. Otherwise, it would randomly select $Q^{\text{round}}$ queries for each bidder and use these to determine his queries. We refer to \refapp{appendix:additional_features} for a more detailed description of this additional design feature. \medskip

\noindent\textbf{Alternatives to VCG payments.}
Recall that our auction computes final payments by applying the VCG
payment rule (see Section~\ref{ssec:CAEfficiency}) to the
bundle-value pairs reported by the bidders (see Line~\ref{alg:MLCA:PaymentComputation} of Algorithm~\ref{alg:MLCA}).
However, this choice is completely modular in our auction, and the auctioneer can instead adopt different payment rules with other desirable properties (e.g., core-selecting payment rules such as \textit{VCG-nearest}~\citep{day2012quadratic}). We provide a more extended discussion on the trade-offs related to the choice of different payment rules in \refapp{appendix:additional_features}.

\section{ Theoretical Analysis}
\label{sec:theoreticalAnalysis}

We now turn to the theoretical analysis of the MLCA mechanism.

\subsection{Learning Error and Elicitation Guarantees}
\label{sec:learningError}

Because MLCA employs an ML algorithm, the effectiveness of MLCA in terms of eliciting the most useful information and its overall efficiency
are related to the \emph{prediction performance} of that ML algorithm.
In this section, we formalize this relationship by providing
theoretical guarantees that can be used by an auctioneer as guidance in
selecting and parameterizing an ML algorithm for use in MLCA.

The prediction performance of an ML algorithm is usually assessed via
its \emph{learning error}.  Formally, for a bidder $i$, we define the
learning error on a bundle $x$ as the absolute value of the difference
between the bidder's true value and the learned value for that bundle,
i.e.: $\text{err}(x) = |v_i(x) - \tilde{v}_i(x)|$.  Of course, in
practice, an auctioneer cannot evaluate the learning error directly
because she does not have access to the true valuations. However, she
can evaluate the ML algorithm using data from previous auctions. In
general, the better the ML algorithm captures the domain, the smaller
the learning errors will be.  In Section~\ref{sec:OptMLExp}, we
evaluate the learning error of four different ML algorithms and show
how they depend on the richness of the ML model used.  Throughout this
section, we leave the ML algorithm unspecificied and derive
theoretical properties that hold for any ML algorithm.

\subsubsection{Bounding the Efficiency Loss}
\label{subsubsec:EfficiencyLossBound}

In service of connecting the learning error of the ML algorithm to the
efficiency loss of MLCA, we start with a bound on the efficiency loss
of the learned valuation $\tilde a$, as follows:

\begin{proposition}\label{prop:learning_error_learned}
Let $\tilde{v}$ be a learned valuation profile. Let $\tilde a$ be an
efficient allocation w.r.t. to $\tilde{v}$, and let $a^*$ be an
efficient allocation w.r.t. the true valuation profile. Assume that
the learning errors in the bundles of these two allocations are
bounded as follows: for each bidder $i$, $|\tilde{v}_i(\tilde a) -
v_i(\tilde a)|\le\delta_1$ and $|\tilde{v}_i(a^*) -
v_i(a^*)|\le\delta_2$, for $\delta_1,\delta_2\in\mathbb R$. Then the
following bound on the efficiency loss in $\tilde a$ holds:
\begin{equation}\label{eq:learning_error_learned}
1 - \eff(\tilde a) \le  \frac{n(\delta_1+\delta_2)}{V(a^*)} .
\end{equation}
\end{proposition}
\begin{proof}
We add and subtract $\displaystyle V(\tilde a)$ and
$\displaystyle \tilde V(a^*)$ to $\displaystyle V(a^*) -
V(\tilde a)$, obtaining
\begin{equation}\label{eq:learning_error_learned2}
V(a^*) - V(\tilde a) = \left(V(a^*) - \tilde V(a^*)\right)
+ \left( \tilde V (\tilde a) - V(\tilde a)\right) +
\left( \tilde V (a^*)-\tilde V (\tilde a)\right).
\end{equation}
Given that $\tilde a$ is a social welfare-maximizing allocation under
$\tilde v$, the term $\displaystyle \tilde V (a^*)-\tilde V ( \tilde
a)$ in \eqref{eq:learning_error_learned2} cannot be
positive. Inequality~\eqref{eq:learning_error_learned} follows by
considering that any real number cannot be greater than its absolute
value, and dividing both side by $V(a^*)$.
\end{proof}

In words, Proposition~\ref{prop:learning_error_learned} shows that the
efficiency loss at $\tilde{a}$ is bounded by the overall learning
error at the two allocations $a^*$ and $\tilde{a}$ (normalized by the
maximum social welfare). Note that the proposition does not require
any assumptions regarding bidders' strategies; in particular, it holds
independent of whether $\tilde{v}$ was learned based on truthful or
misreported values.

Now that we have a bound on $\tilde{a}$, we can derive a bound on the
final allocation $a^{MLCA}$:
\begin{proposition}\label{prop:learning_error}
Assume that all bidders only submit truthful bids to MLCA. Let
$\tilde{v}$ be some learned valuation profile, let $\tilde a$ be an
efficient allocation w.r.t. to $\tilde{v}$, and let $a^*$ be an
efficient allocation w.r.t. the true valuation profile.  Assume that
the learning errors in the bundles of $\tilde{a}$ and $a^*$ are
bounded as follows: for each bidder $i$, $|\tilde{v}_i(\tilde{a}) -
v_i(\tilde{a})|\le\delta_1$ and $|\tilde{v}_i(a^*) -
v_i(a^*)|\le\delta_2$, for $\delta_1,\delta_2\in\mathbb R$. Then
following bound on the efficiency loss in the final allocation
$a^{MLCA}$ holds for all $\tilde{v}$ learned for any economy, in any
iteration of MLCA (i.e., in the Query Module;
Line~\ref{AlgQueryModule:LearningStep} of
Algorithm~\ref{alg:ML_Elicitation}):
\begin{equation}\label{eq:learning_error}
1 - \eff(a^{MLCA}) \le  \frac{n(\delta_1+\delta_2)}{V(a^*)}
\end{equation}
\end{proposition}
\begin{proof}
Given Proposition~\ref{prop:learning_error_learned}, it suffices to
show that $\eff(a^{MLCA})\geq \eff(\tilde{a})$. This follows from
considering that, whenever the Query Module is called, MLCA afterwards
ensures that all bundles contained in $\tilde a$ will be queried. This
implies that, under truthful reports, the efficiency of MLCA can only
be (weakly) higher than the efficiency of $\tilde a$.  Thus,
$\eff(a^{MLCA})\geq \eff(\tilde{a})$.
\end{proof}
\medskip

Proposition~\ref{prop:learning_error} provides a bound on the
efficiency loss of the mechanism as a whole, in the learning error of
the ML algorithm. It follows that, with truthful reports, the mechanism is fully efficient if it has a perfect ML algorithm with zero learning error.

\subsection{Incentives}
\label{SEC:Incentives}

In this section, we study the incentive properties of MLCA. We first explain how MLCA's design features promote good incentives \textit{in practice} (Section~\ref{subsubsec:GoodIncentivesInPractice}). We then show that, under certain assumptions, MLCA is social welfare aligned and ex-post Nash incentive compatible (Section~\ref{subsubsec:SocialWelfareAlignmentAndExPostNash}).

\subsubsection{Design Features Promoting Good Incentives in Practice}
\label{subsubsec:GoodIncentivesInPractice}

For MLCA to achieve high allocative efficiency, it is important that bidders act truthfully. Formally:

\begin{definition}
\label{def:truthful_strategy}
    In MLCA, a bidder’s strategy is called \textit{truthful} if the bidder (1) only replies truthfully to any value query and (2) only pushes truthful value reports.
\end{definition}
\smallskip

Given the iterative design of MLCA, bidders have various direct and indirect ways to influence the final allocation as well as their payments. Without making any assumptions, it is not possible to formally guarantee that it is always optimal for bidders to be truthful. Furthermore, as is the case for all deployed iterative CAs, it is analytically intractable to determine the equilibrium bidding strategies in MLCA (e.g., see \cite{Levin2016PropertiesOfCCA} for an analysis of the incentives of the CCA). Nevertheless, we argue that MLCA has good incentives in practice, in the sense that any possible manipulation strategies are implausible to execute successfully and too risky, such that a rational bidder should simply report truthfully. The following four design features of MLCA promote this behavior:
\begin{itemize}
    \item \textbf{Design Feature \#1:} The MLCA outcome is computed based on the elicited bids only
    \item \textbf{Design Feature \#2:} MLCA computes VCG payments (based on the elicited bids)
    \item \textbf{Design Feature \#3:} MLCA explicitly queries each bidder's marginal economy
    \item \textbf{Design Feature \#4:} MLCA allows for push-bids
\end{itemize}
\smallskip

Design Feature \#1: It is important to note that the final allocation (Step~\ref{alg:MLCA:FinalAllocation} in Algorithm~\ref{alg:MLCA}) and the final payments (Step~\ref{alg:MLCA:PaymentComputation} in Algorithm~\ref{alg:MLCA}) are computed based on the set of elicited bids only. In particular, no ML is used in this step. This implies that bidders only care about the set of bids that MLCA elicits.

Design Feature \#2: Based on these bids, MLCA computes VCG payments.\footnote{See \refapp{appendix:additional_features} for a discussion regarding core payments.} Of course, without requiring bidders to report their full value function, VCG is no longer strategyproof \citep{NisanRonen2000}. To see the different pathways of how a bidder may be able to benefit from misreporting in MLCA, consider a bidder's utility when the VCG payment rule is applied to the bundle-value pairs at the end of MLCA:
\begin{equation}
\label{eq:utility_MLCA}
u_i=\underbrace{\bigg( v_i(a^{\text{MLCA}}_i) + \sum_{j \in N \setminus \{i\}} \hat v_j (a^{\text{MLCA}}_j) \bigg)}_{\text{Reported social welfare of main economy}} \quad - \underbrace{\sum_{j \in N \setminus \{i\}} \hat v_j (a^{-i}_j)}_{\text{Reported social welfare of marginal economy}}
\end{equation}

Observe from this equation that a bidder seeking to manipulate MLCA
can evaluate the effect of a manipulation by grouping
the three terms in his utility function into two components: the first
two terms are the \textit{reported social welfare in the main economy}
(when bidder $i$ is truthful), while the last term is the
\textit{social welfare of bidder $i$'s marginal economy}. Any
beneficial manipulation must increase the difference between the first
and the second component. We next argue that, with Design Features \#3 and \#4 in place, it is practically implausible to successfully execute any such manipulation.

Design Feature \#3: In every auction round, MLCA explicitly queries each bidder's marginal economy (Step~\ref{alg:MLCA:QueryMarginal} in Algorithm~\ref{alg:MLCA}). This means that, for each bidder $i$, the query module generates a new query profile with one query for all other bidders, such that this query profile maximizes the learned social welfare in bidder $i$'s marginal economy. This makes it very likely that, for the final outcome computation, most or even all bids constituting the optimal allocation in bidder $i$’s marginal economy have been elicited without considering his reports. Thus, the only way for bidder $i$ to have an effect on the social welfare of his marginal economy is to affect (via his own bids) how the other bidders' valuations are learned. However, to safely execute such a manipulation, a bidder would essentially need complete information about the other bidders and perfect information about all steps taken by the mechanism. In practice, this is implausible for multiple reasons. First, not only does a bidder have uncertainty regarding the other bidders’ values but also regarding their strategies (e.g., which bundle-value pairs they will push). Second, due to the design of MLCA, a bidder does not know which queries the other bidders have been asked (e.g., due to the random initial bundles; see Remark~\ref{rem:InformationSharingInMLCA}). Given these two sources of uncertainty, we argue that any manipulation of the marginal economy is very risky for a bidder and thus implausible in practice.

If we accept that a bidder's marginal economy is practically independent of his reports, a bidder's only way to increase his utility is to increase the reported social welfare of the main economy (i.e., the first term in Equation~\ref{eq:utility_MLCA}). Here, bidder $i$'s reports are used in the optimization that determines $a^{\text{MLCA}}$. The following example illustrates how a bidder might (theoretically) want to misreport his value to increase the social welfare in the main economy.
 
\begin{example}
\label{ex:manipulation}
Consider a setting with 2 bidders and 2 items. We let
$\textsc{a}=(1,0), \textsc{b}=(0,1), \textsc{ab}=(1,1)$. We use MLCA
with linear regression, which in this setting estimates 2 coefficients
(one per item). The other MLCA parameters are set as follows:
$Q^{\text{max}}=2$, and $Q^{\text{init}}=1$. Bidder $1$'s true values
are $v_1(\textsc{a})=v_1(\textsc{ab})=2$, $v_1(\textsc{b})= 1.1$, and
bidder $2$'s ones are $v_2(\textsc{a})=v_2(\textsc{b})=1$,
$v_2(\textsc{ab})=2$.  Implicitly, MLCA assumes a report from both of
zero for the empty bundle. Assume that, in
Line~\ref{alg:MLCA:InitialReports} of Algorithm~\ref{alg:MLCA}, bidder
$2$ is assigned query \textsc{ab} and he truthfully reports $\hat
v_{2}(\textsc{ab})=2$; bidder~$1$ is assigned query $\textsc{b}$, and
he truthfully report $\hat v_{1}(\textsc{b})=1.1$.

The mechanism then calls the Query Module to determine new queries for
the main economy. The valuations learned in
Line~\ref{AlgQueryModule:LearningStep} of
Algorithm~\ref{alg:ML_Elicitation} are: $\tilde v_1(\textsc{a})=0$ and
$\tilde v_1(\textsc{b})=\tilde v_1(\textsc{ab})=1.1$, and $\tilde
v_2(\textsc{a})=\tilde v_2(\textsc{b})=1$, and $\tilde
v_2(\textsc{ab})=2$.  Bidder $1$ is then assigned
query $\textsc{ab}$, while bidder $2$ is assigned query $\textsc{a}$.
If bidders are truthful, bidder~1 is allocated bundle $\textsc{b}$ and
charged $p_1^{\text{MLCA}} = 1$, with utility equal to $1.1 - 1 =
0.1$.

If bidder~$1$ had reported 0.9 for bundle $\textsc{b}$, i.e., $\hat
v_{1}(\textsc{b})=0.9$, then the valuation $\tilde v_1$ learned in
Line~\ref{AlgQueryModule:LearningStep} of
Algorithm~\ref{alg:ML_Elicitation} would have been $\tilde
v_1(\textsc{a})=0$ and $\tilde v_1(\textsc{b})=\tilde
v_1(\textsc{ab})=0.9$; he would have been assigned query $\textsc{a}$
in the main economy, obtaining this bundle in the final allocation for
a payment equal to $1$. Then, his utility would have been $2-1=1>0.1$.
\end{example}
\smallskip

Design Feature \#4: Due to the uncertainty a bidder faces regarding the state of the mechanism (see Remark~\ref{rem:InformationSharingInMLCA}), the effect of the manipulation shown in Example~\ref{ex:manipulation} would be very
difficult to predict for the bidder, which makes it very risky. But more importantly, the
ability for bidders to push bids in the initial auction phase make the manipulation \emph{unnecessary}, even for bidders who would otherwise contemplate
executing it. Specifically, with push bids, instead of misreporting
his value for $\textsc{b}$ to ultimately be queried and allocated $\textsc{a}$, bidder $1$ could just push his value for $\textsc{a}$, and then report truthfully afterwards. Given this, we argue that any remaining manipulations would be too risky for the bidder to execute and are thus implausible in practice.

\subsubsection{Social Welfare Alignment and Ex-post Nash Incentive Compatibility}
\label{subsubsec:SocialWelfareAlignmentAndExPostNash}

We now show that, given two assumptions, we also obtain formal incentive guarantees for MLCA.

\begin{assumption}\label{assumption:independent_marginal}
	For every bidder $i$, if all other bidders report truthfully, then the social welfare of bidder i’s marginal economy is independent of bidder i’s value reports.
\end{assumption}

In our discussion of Design Feature \#3, we have provided arguments for why this assumption is justified in practice. 
Given this assumption, we can make our earlier arguments regarding a bidder's remaining manipulation opportunities more precise.
\begin{proposition}[Social Welfare Alignment of MLCA]
\label{Prop:SocialWelfareAlignment}
If Assumption~\ref{assumption:independent_marginal} holds, and if all other bidders are truthful, then MLCA is \emph{social welfare aligned}, i.e., increasing the reported social welfare of $a^{\text{MLCA}}$ is the only way for a bidder to increase his utility.
\end{proposition}
\begin{proof}
As we see in Equation~\eqref{eq:utility_MLCA}, in MLCA, the utility of a bidder is given by the difference between the reported social welfare in the main economy and the one in his marginal. Given Assumption~\ref{assumption:independent_marginal}, we know that, when the other bidders are truthful, a bidder cannot affect the reported social welfare of his marginal economy. Thus, the only way for him to increase his utility is to increase the reported social welfare for the final allocation $a^{\text{MLCA}}$.
\end{proof}

Proposition~\ref{Prop:SocialWelfareAlignment} says that any beneficial manipulation for a bidder must increase the reported social welfare in the main economy. However, this is already the primary goal of MLCA's elicitation process. In particular, the query module is designed such that it generates query profiles with maximum learned social welfare. Thus, if the ML algorithm works well, MLCA should already perform this task on behalf of the bidders, without the need for any manipulations. Indeed, in our experiments (Section~\ref{sec:MLCAvsCCAExp}), we show that, when all bidders report truthfully, MLCA finds the efficient allocation in the majority of auction instances, for two of the three domains we study. This motivates the second assumption:

\begin{assumption}
\label{assumption:efficient_main}
If all bidders bid truthfully then MLCA finds an efficient allocation.
\end{assumption}

With this assumption in hand, we can now provide formal incentive guarantees for MLCA:

\begin{proposition} [Truthful Ex-post Nash Equilibrium]
	If Assumptions~\ref{assumption:independent_marginal} and~\ref{assumption:efficient_main} hold, then bidding truthfully is an ex-post Nash equilibrium of the game induced by MLCA.
\end{proposition}
\begin{proof}
	Given Assumption~\ref{assumption:independent_marginal}, we know that MLCA is social welfare aligned; i.e., any beneficial manipulation must increase the social welfare of the main economy. Given Assumption~\ref{assumption:efficient_main}, MLCA finds an efficient allocation if all bidders are truthful. Thus, there is no possibility to further increase the reported social welfare in the main economy. The result follows.
\end{proof}

In practice, Assumptions~\ref{assumption:independent_marginal} and ~\ref{assumption:efficient_main} do not hold perfectly. Therefore, we also analyze the robustness of MLCA against various strategic manipulations via computational experiments (see Section~\ref{subsec:ManipulationExperiments}). In our experiments, we could not find any manipulation strategy that increases a bidder's utility, providing further evidence regarding MLCA's robustness against manipulations in practice.

%

%

\begin{remark}[Collusion and Spitefulness]
A potential remaining concern regarding manipulative behavior is
\emph{collusion} (where two or more bidders manipulate together to
increase each others' profits) and \emph{spiteful bidding} (where a
bidder manipulates to decrease another bidder's profits). Given that
MLCA uses the VCG payment rule, it is (like VCG) also susceptible to
collusive and spiteful behavior. However, note that other practical
CAs like the CCA are also susceptible to similar manipulations
\citep{Levin2016PropertiesOfCCA}. Investigating the degree to which
our mechanism can be manipulated via spiteful or collusive bidding,
and comparing it to the CCA, is outside the scope of this paper but
interesting future work.
\end{remark}

\subsection{Individual Rationality}
\label{sec:IndividualRationality}
A mechanism is \textit{individually rational} if each bidder's payment
is less than or equal to the bidder's reported value for his final
allocation. We can show:
\begin{proposition}
MLCA satisfies individual rationality.
\end{proposition}

\begin{proof}
Note that, despite its iterative nature, MLCA computes its final
payments by just applying a VCG (or a core-selecting) payment rule on
the reported valuation profile $\hat v$ (see
Equation~\eqref{eq:payment_MLCA}, Algorithm~\ref{alg:MLCA}). As each
$\hat v_i$ is only defined on bundles evaluated by bidder $i$ (and
assigns bidder $i$'s reported values to these bundles), MLCA inherits
individual rationality from using the VCG (or core-selecting)
mechanism. More formally, if we expand each bidder $i$'s utility
$\hat v_i(a^{\text{MLCA}})-p_i^{\text{MLCA}}$ using our payment
equation~\eqref{eq:payment_MLCA}, we obtain $\sum_{j\in N} \hat
v_j(a^{\text{MLCA}}) - \sum_{j\in N\setminus\{i\}} \hat v_j
(a^{(-i)})$, which, as $a^{\text{MLCA}}$ is optimal for $\hat v$,
cannot be negative.
\end{proof}
\subsection{No-Deficit}
\label{sec:NoDeficit}

A mechanism stipulating transfers between its participants and the
center should not run a \textit{deficit}. We can show that MLCA
guarantees the \textit{no-deficit} property.
\begin{proposition}
        MLCA satisfies no-deficit.
\end{proposition}

\begin{proof}
As we discussed in the proof above, MLCA computes its final payments
by just applying a VCG (or a core-selecting) payment rule on the
reported valuation profile $\hat v$ (see
Equation~\eqref{eq:payment_MLCA}, Algorithm~\ref{alg:MLCA}). Our
proposition follows by considering that VCG (or core-selecting)
payments are always non-negative for each reported valuation profile
$\hat v$.
\end{proof}

\section{Instantiating the Machine Learning Algorithm}
\label{sec:MLAlgorithm}

So far, our presentation of MLCA has been agnostic about which
ML algorithm to use. However, in practice,
this choice is very important as it determines not only the
quality of the queries identified by
Algorithm~\ref{alg:ML_Elicitation}, but also whether running this
algorithm is computationally feasible in practice. Indeed, each time our Query Module is called, the mechanism may need to determine up to $n+1$
social welfare-maximizing allocations for the learned valuations, each
time under different constraints. The evaluation of each allocation
$a$ requires applying $\mathcal A_i$ for each bidder $i$ to derive
their respective $\tilde v_i$.  In general, explicitly considering each
allocation would require evaluating exponentially many allocations. However, if the learning model exhibits useful structure, we can exploit this structure when searching for the social
welfare-maximizing allocation.
In this section, we present learning models that can be
used in our mechanism and show how to integrate their
predictions into our mechanism's winner determination problem in
a computationally practical way.  Specifically, we provide an integer
program formulation of the winner determination problem used to derive
$\tilde a$ by careful representation of each trained ML
model. Throughout this section, we also assume that each bidder has
made the same number of reports $\ell$.
In Section~\ref{sec:linearRegression}, we start with a simple linear
regression model. While simple to understand and computationally
convenient, this learning model cannot capture complementarities
between items, which prevents Algorithm~\ref{alg:ML_Elicitation} from
learning highly efficient allocations in settings where such
complementarities are important. Then, in Section~\ref{sec:SVR}, we introduce
more \textit{expressive} learning models, i.e., learning models able
to capture more complex valuations.%
\footnote{Most theoretical work in ML assumes that the input data is
  IID.  Despite this, many ML algorithms still work well on dependent
  samples in practice.  Here, we take advantage of this empirical
  efficacy, as the training data in our framework is not independent.
  Of note, in a non-IID context, learning error may not diminish with
  additional training data.  We leave it to future work to incorporate
  techniques specifically built around non-IID data
  \citep[e.g.,]{steinwart2009learning}.}%

\begin{remark}
\label{rem:DifferentMLForDifferentBidders}
In general, MLCA can use $n$ bidder-specific learning algorithms
$\mathcal A_i$ (see Line~\ref{alg:MLCA:Parameters} in
Algorithm~\ref{alg:MLCA}). For simplicity, in this section we
introduce our winner determination formulations under the
assumption that each bidder uses the same learning algorithm for
their learned valuation function $\tilde{v}_i$. We emphasize that
because bidders' ML algorithms are treated independently it is easy to generalize to the case where different ML algorithms $\mathcal A_i$ are used for each bidder.
\end{remark}

\subsection{Linear Regression}\label{sec:linearRegression}
In linear regression, each learned valuation $\tilde v_i$ has the
linear structure
\begin{equation}\label{eq:linearRegression}
\tilde v_i (x) = w_i \cdot x,
\end{equation}
where $w_i \in \mathbb R^m$ is the \textit{weight vector}
corresponding to $\tilde v_i$.%
\footnote{Here we do not introduce a bias term. This is consistent
  with assuming that our bidders always value the empty bundle at zero.}%
Given a weight vector $w_i$, each $w_{ij}$ represents
bidder $i$'s learned value for item $j$.  A first, simple approach to
determine $\tilde{v}_i(x)$ is to interpolate bidder $i$'s reported
values $R_i=\{(x_{ik}, \hat v_{ik})\}_{k \in L}$. We can then use a
standard squared loss function $\mathcal L_2(y,\tilde y) = (y-\tilde
y)^2$ to quantify the interpolation error between an observation $y$
and prediction $\tilde y$, and choose our weight vector as
\begin{equation}\label{eq:NonRegLinPrimal}
w_i \in \argmin_{w_i'} \, \sum_{k \in L} \mathcal{L}(\hat v_{ik} ,
w_i' \cdot x_{ik}).
\end{equation}
Note that this simple approach does not provide a full
characterization of $ w_i$ and, consequently, of bidder $i$'s learned
valuation.  We then add the extra term $||w_i'||^2=w'_i\cdot w'_i$ to
the objective in~\eqref{eq:NonRegLinPrimal}, obtaining a convex
problem. This technique is generally called \textit{Tikhonov
  regularization} \citep{tikhonov1977solutions}, and it biases the
learning towards valuations with low weights $w_{ij}$. This new
learning model is known as \textit{regularized linear regression} and
determines the weight vector as:
\begin{equation}\label{eq:RegLinPrimal}
w_i \in \argmin_{w_i'} \, c\, \sum_{k\in L} \mathcal{L}(\hat v_{ik} ,
w_i' \cdot x_{ik})+||w_i'||^2
\end{equation}
Here $c>0$ defines the tradeoff between interpolation accuracy and
regularization.

When valuations are learned via regularized linear regression, the
optimization problem that determines $\tilde a \in
\argmax_{a\in\mathcal F} \sum_{i\in N}\tilde v_i(a)$ can be formulated
as the following integer program (IP):
\begin{align}
\label{eq:WDLinearRegression}
\max_{a} & \quad \sum_{i\in N} \sum_{j\in M} w_{ij} a_{ij}  \\
\text{s.t. } & \quad \sum_{i\in N} a_{ij}\le 1 \quad  \forall j\in M \notag \\
& \quad a_{ij} \in \{0,1\} \quad \forall i\in N, j\in M. \notag
\end{align}
This IP has $n \cdot m$ boolean variables $a_{ij}$, each denoting
whether bidder $i$ should get item $j$ under allocation $a$, and $m$
feasibility constraints enforcing that each item $j$ is not allocated
more than once.
To limit this search problem to the allocations of the set $\mathcal
F_i$ defined in Line~10 of Algorithm~\ref{alg:ML_Elicitation}, it is
sufficient to add an integer cut $\sum_{j\in M} |x_{ij} - a_{ij}|\ge
1$ for each bundle $x$ in $R_i$.  While integer programming problems
are NP-hard, using branch and bound algorithms
\citep{land1960automatic} as implemented by modern IP solving
software such as CPLEX \citep{cplex} allows us to solve
Problem~\eqref{eq:WDLinearRegression} in a few milliseconds even for
large auction instances with 98 items and 10 bidders.

Despite being computationally very convenient, linear regression
models have a major drawback: they cannot capture valuations where
items are complements or substitutes.\footnote{Under linear regression models, the learned value bidder $i$
  has for bundle $x$ is always additive as it is given by the sum of
  the learned values $ w_{ij}$ for the items $j$ contained in $x$ (see
  Equation~\eqref{eq:linearRegression}).} In the next section, we introduce a learning model that generalizes linear regression and allows us to capture a broader class of
valuations.

\subsection{Support Vector Regression}
\label{sec:SVR}

Support vector regression (SVR) is a learning model that provides
powerful non-linear learning, while retaining attractive computational
properties.  In this subsection, we present the most important
properties of SVR. We defer to \citet{smola2004tutorial} for a
detailed introduction.

To learn non-linear valuations, SVR algorithms project our bundle
encodings from $\mathbb R^m$ into a high (possibly infinite)
dimensional \emph{feature space} $\mathbb R^f$. This projection is
determined via a mapping function $\varphi: \mathbb R^m \to \mathbb
R^f$ that is an implicit meta-parameter of the learning model (as we
shall see, it is specified via the choice of kernel). The learned
function is then captured via linear regression in the feature space
as
\begin{equation}\label{eq:regressionFeatureSpace}
\tilde{v}_i(x) =  w_i \cdot \varphi(x),
\end{equation}
where $w_i \in \mathbb R^f$. The weight vector $ w_i$ minimizes an
objective similar to the one in~\eqref{eq:RegLinPrimal}, with the
important difference that the squared loss is replaced by the
$\varepsilon$-insensitive hinge loss $\mathcal L_\varepsilon(y,\tilde
y) = \max \{|y- \tilde y|-\varepsilon,0\}$, where $\varepsilon \ge 0$
is a meta-parameter of the model.
The $\varepsilon$-insensitive hinge loss is a linear function in
contrast to the squared loss in the regularized least squares approach
described earlier.  Moreover, when the interpolation error is smaller
than $\varepsilon$, the loss drops to zero.  This ``insensitive''
region will turn out to be important in that it allows for more
succinct models to be learned.
In standard machine learning applications, this parsimony is useful
for speeding up the computation of the trained model, $\tilde
v_i(x)$. In our application, the parsimony of models trained under the
$\varepsilon$-insensitive loss translates into making the winner
determination problem easier to solve.

With this background we formulate the learning problem under SVRs as
follows \citep{smola2004tutorial}:
\begin{equation}\label{eq:SVRPrimal}
w_i \in \argmin_{w_i'} \, c\, \sum_{k\in L}
\mathcal{L}_\varepsilon(\hat v_{ik} , w_i' \cdot \varphi(x_{ik}))
+||w_i'||^2.
\end{equation}
When the dimensionality of the feature space is low, $w_i$ can be
determined using standard simplex algorithms, as for linear
regression. However, SVR algorithms may involve high dimensional
feature spaces, meaning that it is often impractical to determine
$\tilde v_i$ via its weight vector. In this scenario, it is convenient
to derive $\tilde v_i$ via the \textit{dual version} of the learning
problem presented in \eqref{eq:SVRPrimal}, which is formulated in
$\ell$ pairs of variables $\alpha_{ik}$ and $\beta_{ik}$; these
variables are the Lagrange multipliers of the two constraints $w_i'
\cdot \varphi(x_{ik}) \le \hat v_{ik}+\varepsilon$ and $w_i' \cdot
\varphi(x_{ik}) \ge \hat v_{ik}-\varepsilon$, respectively.
To do this, we begin by defining the \emph{kernel function}, which is
defined as the value of the dot product of the projection of two
vectors $x$ and $x'$ in the feature space:%
\begin{definition}[Kernel Function \citep{mercer1909functions}]
\begin{equation}
\kappa(x,x')=\varphi(x) \cdot \varphi(x')
\end{equation}
\end{definition}
The kernel is useful because of the so-called ``kernel trick'': for a
suitable choice of $\varphi$, $\kappa(\cdot)$ can be evaluated
directly in closed form, without the need to envoke the explicit
projection $\varphi$ at all.

With this definition, the dual formulation of the SVR training problem,
\eqref{eq:SVRPrimal}, is as follows \citep{smola2004tutorial}:
\begin{align}
\label{eq:SVRDual}
\max_{\alpha_i,\beta_i} & \quad -\frac{1}{2} \sum_{k \in L} \sum_{k' \in L} (\alpha_{ik}-\beta_{ik})(\alpha_{ik'}-\beta_{ik'})\, \kappa(x_{ik},x_{ik'}) \\
& \quad -\varepsilon \sum_{k\in L} (\alpha_{ik}+\beta_{ik}) + \sum_{k\in L} \hat v_{i,k}(\alpha_{ik}-\beta_{ik}) \notag \\
\text{s.t. } & \quad \alpha_{ik},\beta_{ik} \in [0,c] \quad  \forall k\in L .\notag
\end{align}
Because of the kernel trick, for a suitable kernel, this problem can
be formulated and solved without entering the high dimensional feature
space at all.  Solving the training problem is thus equivalent to
determining the optimal $\alpha_i$ and $\beta_i$.  Having accomplished
this, predictions under the learned model are obtained via the
following \citep{smola2004tutorial}:
\begin{equation}
\label{eq:SVRDualPrediction}
\tilde{v}_i(x) = \sum_{k \in L} (\alpha_{ik} - \beta_{ik}) \kappa(x, x_{ik})
\end{equation}

\subsubsection{Formulating the Winner Determination Problum Under SVR}

In our application, we need to solve for the social welfare maximizing
allocation under market clearing constraints.  That is, we need to be
able to formulate the objective:
\begin{equation}
\label{eq:socialwelfareobjective}
\tilde a \in \argmax_{a\in\mathcal F} \sum_{i\in N}\tilde v_i(a)
\end{equation}
For a given kernel, we can formulate \eqref{eq:socialwelfareobjective}
via \eqref{eq:SVRDualPrediction}, as follows:
\begin{align}\label{eq:WDSVRGeneric}
\max_a & \quad \sum_{i\in N} \sum_{k\in L} ( \alpha_{ik} -  \beta_{ik}) \kappa(a_i, x_{ik})\\
\text{s.t.} & \quad \sum_{i\in N} a_{ij}\le 1 \quad  \forall j\in M \notag \\
& \quad a_{ij} \in \{0,1\} \quad \forall i\in N, j\in M \notag.
\end{align}
It remains, however, to pick a kernel function $\kappa$ that is
non-linear (to gain expressive power over regularized linear
regression) and that can still be effectively encoded in an integer program.

In selecting kernels, we focus on two commonly used classes:
\textit{dot-product kernel functions} and \textit{radial basis kernel
  functions (RBF kernels)} \citep[][chapter 4]{williams2006gaussian}.
In \refapp{app:WDs}, we present the integer programming
formulations of Problem~\eqref{eq:WDSVRGeneric} for dot-product and
RBF kernels.  In our application, we consider three kernel functions
selected based on two criteria: 1. the expressivity of the
corresponding learning model (which is determined by the implicit
feature mapping $\varphi$), and 2. the complexity of the corresponding
instantiation of Problem~\eqref{eq:WDSVRGeneric}.

We start by considering the \textit{Quadratic kernel}:
\begin{definition}[Quadratic Kernel]
\label{def:kernelQuadratic}
A Quadratic kernel is a kernel function of the form
\begin{equation}
\kappa(x,x') = x\cdot x' + \lambda (x\cdot x')^2,
\end{equation}
where $\lambda$ is a non-negative parameter.
\end{definition}
The Quadratic kernel is not \textit{fully expressive}, i.e., it cannot
capture \textit{any} valuation function a bidder may have. However, it
allows us to formulate Problem~\eqref{eq:WDSVRGeneric} as a quadratic
programming problem with boolean variables, which can be practically
solved via branch and bound methods \citep{lima2017solution}:
\begin{align}\label{eq:WDSVRQuadratic}
\max_a & \quad \sum_{i\in N} \sum_{k\in L} ( \alpha_{ik} - \beta_{ik}) \left(\sum_{j\in M} a_{ij} x_{ikj} + \lambda \left(\sum_{j\in M} a_{ij} x_{ikj}\right)^2\right)\\
\text{s.t.} & \quad \sum_{i\in N} a_{ij}\le 1 \quad  \forall j\in M \notag \\
& \quad a_{ij} \in \{0,1\} \quad \forall i\in N, j\in M \notag.
\end{align}

We also consider two fully expressive kernels: the \textit{Gaussian
  kernel} and the \textit{Exponential kernel}:
\begin{definition}[Gaussian Kernel]\label{def:kernelGaussian}
A Gaussian kernel is a kernel function of the form
\begin{equation}
\kappa(x,x') = \exp{\left(-\frac{\displaystyle ||x -
    x'||^2}{\displaystyle \lambda}\right)},
\end{equation}
where $\lambda$ is a non-negative parameter.
\end{definition}
\begin{definition}[Exponential Kernel]\label{def:kernelExponential}
An Exponential kernel is a kernel function of the form
\begin{equation}
\kappa(x,x') = \exp{\left(\frac{\displaystyle x\cdot x'}{\displaystyle
    \lambda}\right)},
\end{equation}
where $\lambda$ is a non-negative parameter.
\end{definition}
The \textit{Gaussian kernel} is an RBF kernel commonly used in the
machine learning literature. The Exponential kernel is also fully expressive
but, unlike the Gaussian one, is a dot-product kernel.

\section{Optimizing the Machine Learning Algorithm}
\label{sec:OptMLExp}

To achieve maximal performance with MLCA, we need to identify the kernel and parameters that work best for the machine learning algorithm.  Before evaluating the full MLCA in the next section, we here perform a series of experiments to identify the optimal parameterization.
\subsection{Experiment Set-up}
\label{sec:OptMLExpSetup}

In order to run these experiments, we need data with which to exercise
the mechanism.
The allocation of spectrum is one of the most important applications
of CAs.  We therefore adopt the allocation of spectrum for our
experimental evaluation.  To do this, we employ the Spectrum Auction
Test Suite (SATS) version 0.7.0 \citep{weiss2017sats}, which allows us
to easily generate thousands of auction instances on demand. 
We tested our approach on three of SATS' value models
across a range of complexity.  We describe each in turn below:
\begin{itemize}
\item The Global Synergy Value Model (GSVM)
  \citep{goeree_etal_2008_HierarchicalPackageBidding} generates
  medium-sized instances with 18 items and 7 bidders.  GSVM models the
  items (spectrum licenses) as being arranged in two
  circles. Depending on his type, a bidder may be interested in
  licenses from different circles and has a value that depends on the
  total number of licenses of interest.  GSVM also includes a
  ``global'' bidder with interest in two thirds of the total licenses.
\item The Local Synergy Value Model (LSVM)
  \citep{scheffel_etal_2012_OnTheImpactOfCognitiveLimits} also
  generates medium-sized instances with 18 items and 6 bidders. In
  LSVM, the items are placed on a two-dimensional grid, and a bidder's
  value depends on a sigmoid function of the number of contiguous
  licenses near a target item of interest.  This makes the value model
  of LSVM more complex than that of GSVM.\item The Multi-Region Value
  Model (MRVM) \citep{weiss2017sats} generates large instances with 98
  items and 10 bidders.  MRVM captures large settings, such as US and
  Canadian auctions, by modeling the licenses as being arranged in
  multiple regions and bands. Bidders' values are affected by both
  geography and frequency dimensions of the licenses. 
\end{itemize}
The GSVM, LSVM, and MRVM instances we tested correspond to SATS seeds
101-200 for experiments based on 100 samples, and SATS seeds 101-150
for experiments based on 50 samples. We used CPLEX 12.10 to solve the integer programs used to determine welfare-optimal allocations and train our learning algorithms.\footnote{Experimentally we need to evaluate several different mechanism design choices, and do so with enough data points for statistical
significance.  Consequently, in our experiments, we ran an enormous
number of MIPs (approximately $2,000,000$).  Accordingly, we adopt a
modest timeout for the solver, which we set to 1 minute, and adopt the
best solution found so far.  We note that in practical use,
auctioneers will typically have more time to let the optimizer run
(typically at least an hour), which would improve the outcome of our
mechanisms; we thus report conservative results with respect to the
optimality of the MIP solutions.} We conducted our experiments on a
Ubuntu 16.04 cluster with AMD EPYC 7702 2.0 GHz processors using 8
cores and 32 GB of RAM.

\subsection{Results}
\label{sec:OptMLResults}

In the experiments in this section, we are interested in identifying
which kernel to use in the ML algorithm such that the full mechanism
will yield the most efficient outcome.
In selecting a kernel, we need to optimize the various parameters of
the learning algorithm (i.e. $\varepsilon$, $c$, and any kernel
hyperparameters).  We tune these parameters to maximize the effeciency
of the predicted allocation.
Because we need the machine learning algorithm to operate within a
reasonable comptuational budget, one of the most important parameters
is the insensitivity threshold $\varepsilon$.  This is because, as is
standard in kernel-based ML methods, $\varepsilon$ controls the number
of support vectors that are likely to be part of the learned model.
The size of the winner determination MIP that we solve (see
\refapp{app:WDs}) is heavily dependent on the number of support
vectors. Thus, $\varepsilon$ determines a trade-off between the learning
error of the ML model and the run-time of the winner determination
optimizer.  To investigate this trade-off, we conducted an experiment
on the two more expensive kernels, Gaussian and Exponential. We
present the corresponding results for the LSVM domain in
Table~\ref{fig:insensitivityComparisonLSVM}. Results for the other two
domains are provided in
\refapp{app:kernelAndInsensitivityComparisonGSVMandMRVM}.  To
conduct this experiment, we first generate an LSVM domain using SATS.
From this domain we then randomly sample $Q \in \{100,200,500\}$
truthful bundle-value pairs from each bidder.  Based on these reports,
we train an ML algorithm (with particular parameters) to construct
$\tilde{v}$.  We then evaluate this $\tilde{v}$ and its effectiveness
in learning optimal allocations (in terms of quality and speed)
according to four measures of interest which we next describe.

\begin{table}
        \centering

\begin{tablesizeadjustment}
\setlength{\tabcolsep}{1.5pt}%
\begin{tabular}{l | c || r | r | r || r | r | r || r | r | r || r | r | r}
Kernel & $\varepsilon$ & 
\multicolumn{3}{c||}{Efficiency} &
\multicolumn{3}{c||}{Learning Error} & 
\multicolumn{3}{c||}{WD Solve Time} & 
\multicolumn{3}{c}{Optimality Gap} \\
\hhline{~|~||---||---||---||---}
&
& \multicolumn{1}{c|}{50} & 
  \multicolumn{1}{c|}{100} & 
  \multicolumn{1}{c||}{200}
& \multicolumn{1}{c|}{50} & 
  \multicolumn{1}{c|}{100} & 
  \multicolumn{1}{c||}{200}
& \multicolumn{1}{c|}{50} & 
  \multicolumn{1}{c|}{100} & 
  \multicolumn{1}{c||}{200}
& \multicolumn{1}{c|}{50} & 
  \multicolumn{1}{c|}{100} & 
  \multicolumn{1}{c}{200} \\

\hhline{=:=*{4}{::===}}

Exponential &  0 &

          81.7\%    &

          84.2\%   &

          84.6\%   &

              16.08   &

              13.39   &

              11.46   &

          54.39s   &

          60.00s   &

          60.00s   &

          0.49   &

          2.04   &

          6.70   \\

\hhline{-|-*{4}{||---}}

Exponential &  2 &

          82.7\%    &

          84.7\%   &

          84.7\%   &

              16.37   &

              13.51   &

              11.54   &

          47.86s   &

          60.00s   &

          60.00s   &

          0.32   &

          1.56   &

          5.31   \\

\hhline{-|-*{4}{||---}}

Exponential &  4 &

          82.7\%    &

          83.8\%   &

          84.7\%   &

              16.80   &

              13.76   &

              11.77   &

          43.60s   &

          60.00s   &

          60.00s   &

          0.17   &

          1.17   &

          4.28   \\

\hhline{=:=*{4}{::===}}

Gaussian &  0 &

          72.5\%    &

          83.2\%   &

          85.9\%   &

              19.41   &

              15.69   &

              12.61   &

          60.00s   &

          60.00s   &

          60.00s   &

          0.68   &

          1.79   &

          4.44   \\

\hhline{-|-*{4}{||---}}

Gaussian &  2 &

          71.0\%    &

          84.3\%   &

          87.3\%   &

              19.70   &

              15.82   &

              12.79   &

          60.00s   &

          60.00s   &

          60.00s   &

          0.59   &

          1.51   &

          3.71   \\

\hhline{-|-*{4}{||---}}

Gaussian &  4 &

          72.3\%    &

          81.4\%   &

          87.6\%   &

              20.11   &

              16.05   &

              13.10   &

          60.00s   &

          60.00s   &

          60.00s   &

          0.47   &

          1.28   &

          3.14   \\

\end{tabular}
\end{tablesizeadjustment}

        \caption[LSVMkernel]{The effect of varying the insensitivity parameter
                $\varepsilon$ on the predictive performance: both efficiency
                of the predicted optimal allocation, and learning error.
                Results are shown for the two expensive kernels,
                Exponential and Gaussian, in the LSVM domain.\footnotemark}
        \label{fig:insensitivityComparisonLSVM}
\end{table}
\footnotetext{Dashes in the table indicate entries where we could not
  get complete results because CPLEX could not find any feasible
  solution within the given time limit. Note that these are
  preliminary results subject to change due to ongoing updates to our
  code base.}

First, we provide the efficiency of the learned optimal allocation.
To do this, we compute the social welfare of the gold standard optimal
allocation at true values, $V(a^*)$, using the concise MIP formulation
built into SATS, and compare this to the true social welfare of the
learned optimal allocation, $V(\tilde a)$.  Note that this is not the
efficiency of MLCA (which we investigate in the next section), but the
efficiency of the learned optimal allocation $\tilde{a}$ after being
trained on $Q \in \{100,200,500\}$ randomly selected bundles.  The
table includes subcolumns for each $Q$, enabling us to compare the
effect of training data size on the learning algorithm.
For each kernel, we show the results for three insensitivity
thresholds, $\varepsilon$: the one that is best for efficiency on
average across the different $Q$s, one twice that size, and zero.
From the table, we can see that most of the efficiency gains occur in
the first half of the considered range of $\varepsilon$.  Further, we
see that for sufficiently large $\varepsilon$, efficiency is
monotonically increasing in the sample size $Q$.

Next, we report the learning error of the machine learning algorithm.
Here we are measuring learning error in the standard way by measuring
the average absolute difference between the predicted and true value
for all bundles in the domain.\footnote{This is possible for GSVM and
  LSVM which have 18 items and are enumerable in this way; for MRVM we
  compute a similar statistic by sampling $100,000$ bundles.}
From the table we see that the higher the epsilon (and thus the fewer
support vectors) the worse the learning error.

Next, we report the solve time for the winner determination MIP that
finds $\tilde a$ based upon the trained ML model.
From the table we see that the solver always times out unless the
insensitivity threshold $\varepsilon$ is sufficiently large.

Finally, we list the optimality gap reported by the solver when it
stops.  Specifically, this is calculated as
$(\overline{o}-\underline{o})/\underline{o}$, where $\overline{o}$,
and $\underline{o}$ are the solver's proven upper and lower bounds on
the objective value respectively.  When the value is zero, the solver
has proven optimality.
As expected, the optimality gap is closely correlated with the solve
time, but it provides a quantitative measure of the consequence of the
solver stopping early.

Overall from the table, we see that it is important to select epsilon
carefully in order to properly trade off learning performance (as
measured via learning error) with solve time (as measured via WD solve
time) to maximize the efficiency of the learned optimal
allocation (as measured by Efficiency).

Now that we have selected the best parameters for each kernel, we turn
to a head-to-head comparison of each of the kernels we have
introduced.
Accordingly, in each of our domains of study we run the same
experiment as described above with the Linear, Quadratic, Exponential
and Gaussian kernels.  We present the results for the LSVM domain in
Table~\ref{fig:kernelComparisonLSVM}; results for the other two
domains are available in
\refapp{app:kernelAndInsensitivityComparisonGSVMandMRVM}.
Following the result in the previous table, all parameters
(i.e. $\varepsilon$, $c$, and any kernel hyperparameters) in this
kernel comparison are tuned to maximize efficiency.

\begin{table}
        \centering

\begin{tablesizeadjustment}
\setlength{\tabcolsep}{1.5pt}%
\begin{tabular}{l || r | r | r || r | r | r || r | r | r || r | r | r}
Kernel & 
\multicolumn{3}{c||}{Efficiency} & 
\multicolumn{3}{c||}{Learning Error} & 
\multicolumn{3}{c||}{WD Solve Time} & 
\multicolumn{3}{c}{Optimality Gap} \\
\hhline{~||---||---||---||---}
& \multicolumn{1}{c|}{50} & 
  \multicolumn{1}{c|}{100} & 
  \multicolumn{1}{c||}{200}
& \multicolumn{1}{c|}{50} & 
  \multicolumn{1}{c|}{100} & 
  \multicolumn{1}{c||}{200}
& \multicolumn{1}{c|}{50} & 
  \multicolumn{1}{c|}{100} & 
  \multicolumn{1}{c||}{200}
& \multicolumn{1}{c|}{50} & 
  \multicolumn{1}{c|}{100} & 
  \multicolumn{1}{c}{200} \\
\hhline{=*{4}{::===}}
Linear & 70.6\% & 
         76.4\% & 
         78.6\% &
         
             23.12 &
             20.08 &
             18.85 &
        
         0.00s &
         0.00s & 
         0.00s &
         0.00 & 
         0.00 & 
         0.00 \\
\hhline{-*{4}{||---}}
Quadratic & 86.2\% & 
         90.6\% & 
         92.6\% &
         
             19.50 &
             17.01 &
             15.13 &
        
         0.25s & 
         0.45s & 
         0.60s &
         0.00 & 
         0.00 & 
         0.00 \\
\hhline{-*{4}{||---}}
Exponential &          
          82.7\%   &
         
          84.7\%   &
         
          84.7\%   &

             16.37   &
            
             13.51   &
            
             11.54   &

          47.86s   &
         
          60.00s   &
         
          60.00s   &
          
          0.32   &
          
          1.56   &
         
          5.31   \\
\hhline{-*{4}{||---}}
Gaussian &
         
          71.0\%   &
         
          84.3\%   &
         
          87.3\%   &

             19.70   &
            
             15.82   &
            
             12.79   &

          60.00s   &
         
          60.00s   &
         
          60.00s   &
          
          0.59   &
          
          1.51   &
         
          3.71   \\
\end{tabular}
\end{tablesizeadjustment}

        \caption[LSVMkernel2]{Comparison of different kernels in the LSVM domain.  The
                Quadratic kernel obtains the best efficiency under the imposed
                computational constraints.}
        \label{fig:kernelComparisonLSVM}
\end{table}
From the table we can see the the Quadratic kernel yields the best
efficiency at all sample sizes.  The Linear kernel has the lowest
efficiency, showing that more expressive kernels can be worth the
computational effort they require.

Turning to the learning error column, we observe that all of the
entries in this table are higher than the lowest observed value in
Table~\ref{fig:insensitivityComparisonLSVM}, which is the Exponential
kernel with $\varepsilon=0$.  This indicates that the Exponential
kernel is best able to generalize in this domain, but its calculation
is sufficiently expensive that we are generally better off choosing a
Quadratic kernel that has somewhat higher learning error, but a more
succinct winner determination formulation.

We note that efficiency is typically inversely related to learning
error.  For example, the efficiency of the Quadratic kernel rises in
the sample size $Q$, while its learning error decreases in $Q$.  

Turning to the final columns in the table, we see that the Linear
kernel is solvable very rapidly with $0$ optimality gap.  Whereas the
Exponential and Gaussian kernels are more expensive.

We see that similar arguments hold for GSVM and MRVM (see \refapp{app:kernelAndInsensitivityComparisonGSVMandMRVM}). Overall, the Quadratic kernel makes the best tradeoff between learning
error and solve time, yielding the most efficient learned optimal
allocation within our computational budget, and we therefore adopt it
going forward in the experiments in the next section.

\section{Experiments}
\label{sec:MLCAvsCCAExp}

In this section, we evaluate the full MLCA mechanism, comparing it
against the widely-used CCA.  Additionally, we investigate the effect of
non-truthful bidding.

\subsection{Experiment Set-up}
\label{sec:MLCAvsCCAExpSetup}

To run our experiments, we need to specify a number of attributes of
both the MLCA and CCA mechanisms, as well as define the scope and
properties of the experiments themselves.

\paragraph{MLCA.}

MLCA is parametrized by a machine learning algorithm $\mathcal A_i$
for each bidder $i$, the number of queries $Q^{\text{max}}$ and the
number of initial queries $Q^{\text{init}}$.
As described in Section~\ref{sec:OptMLExp}, our ML algorithm will be
an SVR with Quadratic kernel for all domains of study.
In our experiments, we use $Q^{\text{max}}=100$ in the
simpler GSVM domain, and $Q^{\text{max}}=500$ in LSVM and MRVM.
We note that the $Q^{\text{max}}=500$ is a practical
number of queries, given that it was used in the real-world Canadian
auction that inspired the MRVM model that we use.
For the number of initial queries $Q^{\text{init}}$, we select an
optimal value through offline tuning.  Specifically, we use
$Q^{\text{init}} = 50$ for GSVM, $Q^{\text{init}} = 60$ for LSVM, and
$Q^{\text{init}} = 50$ for MRVM.  To be conservative, we don't allow
any bidders to ``push'' bundles in our MLCA evaluation.  This avoids
confusing the results by including contributions due to a bidder push
heuristic.
\paragraph{CCA.}

The CCA is parameterized by the reserve prices employed, the way
prices are updated, and what heuristics are assumed for bidder
behavior in the supplementary phase.
In our experiments, we use reserve prices for each license equal to
1\% of the average license value derived from $10,000$ bundle-value
pairs sampled from bidders in the domain.
For the price updates, we follow the parameterization of the
real-world Canadian CCA, that used $5\%$ price increments for most of
the auction rounds. We note that the number of auction rounds we
obtain under these price settings is similar to the number of auction
rounds observed in the real-world Canadian auction.
Finally, we specify how bidders in the CCA select which bundle-value
pairs to report in the supplementary round.\footnote{There is no prior
  work to guide the optimal strategy for bidders in choosing bundles
  to bid upon in the supplemental round. We therefore explored several
  different heuristics in our experiments after consulting with
  industry experts who have been involved in the design of the CCA and
  who have provided advice to bidders in the CCA. The lack of
  theoretical guidance here represents an additional strategic burden
  on bidders, in contrast to MLCA.}  We consider the following heuristics:
\begin{itemize}
\item Clock Bids: this corresponds to there being \textit{no
}supplementary round. Thus, the final allocation of the CCA is only
  determined based on the those bundles reported in the clock phase
  (as an answer to the corresponding demand query), using the value
  equal to the highest quoted price for that bundle.
\item Clock Bids Raised: bidders provide their true values for all
  unique bundles they reported during the clock phase.
\item Profit Max: bidders report their true values for all bundles
  reported during the clock phase and additionally for $Q$ bundles
  earning them the highest profit at the final clock
  prices.\footnote{To get a similar number of reported values across
    small and large bidders, we also let bidders report values for
    bundles earning them negative profit at the final clock prices,
    which may still be useful for them in the final winner
    determination.} In our simulations, we use $Q=100$ in the
  simpler GSVM domain, and $Q=500$ in LSVM and MRVM\footnote{Our
    implementation of the CCA in MRVM includes \emph{generics},
    enabling bidders to submit ``quantity bids'' for groups of
    substitutable items \citep[see][]{weiss2017sats}. Note that MLCA
    currently does use generics, which gives the CCA a slight
    advantage in this respect.}.
\end{itemize}

\noindent For both MLCA and CCA, we simulate truthful bidding in Section~\ref{sec:MLCAvsCCAResults}; in Section~\ref{subsec:ManipulationExperiments} we will consider more complex bidder behavior.

\subsection{Efficiency Results}
\label{sec:MLCAvsCCAResults}
We now study the efficiency of MLCA, comparing it against the CCA. We  consider each of our three domains in turn. We begin with the
two more stylized domains GSVM and LSVM, to build intuition for when
MLCA works well and when it does not, before we move on to the
realistically-sized MRVM domain.\footnote{Our experiments are run on
  the same computational grid as was used for the experiments in
  Section~\ref{sec:OptMLExp}.  However, the experiments in this
  section required a much larger computational effort because many
  more MIPs needed to be solved in an iterative fashion (more than
  $100,000$ core hours).}

\subsubsection{GSVM}

\begin{table}[H]
    \begin{adjustbox}{max width=\textwidth}
        \centering 

\begin{tablesizeadjustment}
\setlength{\tabcolsep}{1.5pt}%
\begin{tabular}{c | c | c || r | r | r | r }
\multicolumn{1}{c|}{\makecell[t]{Mechanism}} &
\multicolumn{1}{c|}{\makecell[t]{ML Algorithm}} &
\multicolumn{1}{c||}{\makecell[t]{Bidding Heuristic}} &
\multicolumn{1}{c|}{\makecell[t]{Efficiency}} & 
\multicolumn{1}{c|}{\makecell[t]{Revenue}} & 
\multicolumn{1}{c|}{\makecell[t]{Revenue (Core)}} & 
\multicolumn{1}{c}{\makecell[t]{Rounds}} \\ 
\hhline{=:=:=::=:=:=:=}
\multirow{2}{*}{MLCA} & SVR-Linear         & \multirow{2}{*}{-} & 99.7\% (0.10)         & 66.1\% (0.99)   & 69.6\% (0.96)   & 14 (0.0)          \\ 
                      & SVR-Quadratic      &                    & \textbf{100.0\% (0.00)}  & 68.4\% (1.10)   & 72.4\% (1.05)   & 14 (0.0)          \\ 
\hhline{---||----}
\multirow{3}{*}{CCA}  & \multirow{3}{*}{-} & Clock Bids         & 88.8\% (0.81)        & 37.9\% (1.43)    & 52.5\% (0.83)   & 233 (3.0)      \\ 
                      &                    & Clock Bids Raised  & 94.0\% (0.47)          & 51.3\% (1.70)    & 67.1\% (0.93)   & 234 (3.0)      \\ 
                      &                    & Profit Max         & \textbf{100.0\% (0.00)}   & 68.1\% (1.13)   & 73.1\% (0.99)   & 234 (3.0)      \\ 
\hhline{---||----}
VCG                   & \makecell[c]{-}    & \makecell[c]{-}    & \textbf{100.0\% (0.00)} & 68.4\% (1.11)   & \makecell[c]{-} & 1 (0.0)           \\ 
\hhline{---||----}
Random Allocation     & \makecell[c]{-}    & \makecell[c]{-}    & 19.6\% (0.79)        & \makecell[c]{-} & \makecell[c]{-} & \makecell[c]{-} \\ 
\end{tabular}
\end{tablesizeadjustment}

    \end{adjustbox}
        \caption{Results
        for MLCA, CCA and VCG in the GSVM domain. We use
        $Q^{\text{max}}=100$ in MLCA and $Q=100$ in the Profit Max
        heuristic of CCA. Results are averages over 100 auction
        instances. Standard errors are in
        parentheses.}  \label{fig:mechanismComparisonGSVM}
\end{table}

We start with the relatively simple GSVM domain, for which we present
results in Table~\ref{fig:mechanismComparisonGSVM}. We see that MLCA
provides $100\%$ efficiency even with only a $100$ query cap. We note
that bidder preferences in the GSVM domain can be captured perfectly
by the Quadratic kernel (see \refapp{app:GSVMAndQuadKernel}).

The CCA with the Profit Max heuristic also performs very well,
achieving $100\%$ efficiency.
In the table, we next report revenue, measured as the fraction of
surplus accruing to the seller. We see that VCG produces a high level
of revenue in this setting, and that all of the CCA and MLCA versions
are close to VCG, except for the Clock Bids heuristic.
In the next column we show revenue for the VCG-nearest core-selecting
rule as applied to the bundle-value pairs in the supplemental round of
the CCA (this is the payment rule that is often used in practice).  As
expected, we see some revenue lift when swapping VCG for the
core-selecting rule.
Finally the table shows the number of rounds employed by each
mechanism. We observe that MLCA uses a very small number of rounds
given the relatively small query cap.

\subsubsection{LSVM}

\begin{table}[H]
\begin{adjustbox}{max width=\textwidth}
\centering 

\begin{tablesizeadjustment}
\setlength{\tabcolsep}{1.5pt}%
\begin{tabular}{c | c | c || r | r | r | r }
\multicolumn{1}{c|}{\makecell[t]{Mechanism}} &
\multicolumn{1}{c|}{\makecell[t]{ML Algorithm}} &
\multicolumn{1}{c||}{\makecell[t]{Bidding Heuristic}} &
\multicolumn{1}{c|}{\makecell[t]{Efficiency}} & 
\multicolumn{1}{c|}{\makecell[t]{Revenue}} & 
\multicolumn{1}{c|}{\makecell[t]{Revenue (Core)}} & 
\multicolumn{1}{c}{\makecell[t]{Rounds}} \\ 
\hhline{=:=:=::=:=:=:=}
\multirow{2}{*}{MLCA} & SVR-Linear           & \multirow{2}{*}{N/A} & 98.3\% (0.40)          & 72.8\% (0.88)   & 75.3\% (0.79)  & 114 (0.0)         \\ 
                      & SVR-Quadratic        &                      & 99.6\% (0.12)          & 80.9\% (0.91)   & 84.5\% (0.83)  & 114 (0.0)
                                                                                                                                                   \\ 
\hhline{---||----}
\multirow{3}{*}{CCA}  & \multirow{3}{*}{N/A} & Clock Bids           & 82.4\% (0.71)          & 62.8\% (0.98)   & 66.4\% (0.81)   & 123 (0.3)      \\ 
                      &                      & Clock Bids Raised    & 91.0\% (0.47)            & 76.0\% (1.03)     & 79.2\% (0.89)   & 124 (0.3)      \\ 
                      &                      & Profit Max           & \textbf{99.9\% (0.03}) & 82.3\% (0.91)   & 86.4\% (0.73)   & 124 (0.3)      \\ 
\hhline{---||----}
VCG                   & \makecell[c]{-}      & \makecell[c]{-}      & 100.0\% (0.00)              & 83.1\% (0.89)   & \makecell[c]{-} & 1 (0.0)          \\ 
\hhline{---||----}
Random Allocation     & \makecell[c]{-}      & \makecell[c]{-}      & 20.4\% (0.64)          & \makecell[c]{-} & \makecell[c]{-} & \makecell[c]{-} \\ 
\end{tabular}
\end{tablesizeadjustment}

\end{adjustbox}
\caption{Results for MLCA, CCA and VCG in the LSVM domain. We use
$Q^{\text{max}}=500$ in MLCA and $Q=500$ in the Profit Max heuristic
of CCA. Results are averages over 100 auction instances. Standard
errors are in parentheses.} \label{fig:mechanismComparisonLSVM}
\end{table}
Next we turn to the LSVM domain, the results for which we show in
Table~\ref{fig:mechanismComparisonLSVM}. We see that the CCA obtains
an efficiency of 99.9\%. Note that the strong performance of
the CCA is heavily dependent on the employed Profit Max heuristic for
the supplementary phase, as we can see from the much lower performance
of the CCA when using the other two heuristics. MLCA again performs
best when using the Quadratic kernel, achieving an efficiency of
99.7\%, but this is still slightly worse than the performance of the
CCA with the Profit Max heuristic. 
Using a paired t-test we found this difference to be statistically significant ($p<0.01$). We can make similar arguments to GSVM regarding revenue and rounds.

Overall, in the LSVM domain, we observe that the CCA slightly
outperforms MLCA.  This is likely due to two effects. First, most
bidders in the LSVM domain (in particular, the regional bidders) are
only interested in a relatively small number of bundles
\citep{scheffel_etal_2012_OnTheImpactOfCognitiveLimits}. The Profit
Max heuristic with a query cap of 500 then essentially allows them to
almost fully described their preferences. In contrast, MLCA does not
rely on such a strong bidding heuristic, and the Quadratic kernel
employed cannot capture the domain very well
(in contrast to GSVM), which explains the slightly lower efficiency.

\begin{remark}
There are multiple ways to further increase the efficiency of
MLCA. First, with more computational resources (i.e., larger compute
clusters, improvements in compute technology expected to come), the 1
minute time limit we impose on each MIP becomes less of a constraint,
which automatically increases efficiency. Second, with more
computational resources, we could then also employ more expressive
kernels (e.g., Gaussian or Exponential). As we have shown in
Section~\ref{sec:OptMLExp}, with our current computational set-up, the
Quadratic kernel leads to the highest efficiency, even though the
Gaussian and Exponential kernel have lower learning error -- but this
would likely change once we have sufficiently powerful
computers. Third, we could explore using other ML algorithms that may
capture the structure of the LSVM domain better than the Quadratic
kernel. In fact, in work subsequent to this paper,
\cite{Weissteiner2020DeepLearning} have recently shown how deep neural
networks (NNs) can be used as the ML algorithm in MLCA (instead of
SVRs). For the LSVM domain, they showed that using NNs increases the
efficiency of MLCA beyond that achievable via SVRs.
\end{remark}

\subsubsection{MRVM}

\begin{table}[H]
        \centering
        \begin{adjustbox}{max width=\textwidth}

\begin{tablesizeadjustment}
\setlength{\tabcolsep}{1.5pt}%
\begin{tabular}{c | c | c || r | r | r | r }
\multicolumn{1}{c|}{\makecell[t]{Mechanism}} &
\multicolumn{1}{c|}{\makecell[t]{ML Algorithm}} &
\multicolumn{1}{c||}{\makecell[t]{BiddingHeuristic}} &
\multicolumn{1}{c|}{\makecell[t]{Efficiency}} & 
\multicolumn{1}{c|}{\makecell[t]{Revenue}} & 
\multicolumn{1}{c|}{\makecell[t]{Revenue (Core)}} & 
\multicolumn{1}{c}{\makecell[t]{Rounds}} \\ 
\hhline{=:=:=::=:=:=:=}
\multirow{2}{*}{MLCA} & SVR-Linear           & \multirow{2}{*}{N/A} & 93.9\% (0.18)          & 42.7\% (0.33)   & 42.9\% (0.32)   & 114 (0.0)         \\ 
                      & SVR-Quadratic        &                      & \textbf{96.4\% (0.13)} & 40.5\% (0.33)   & 40.7\% (0.34)   & 114 (0.0)         \\ 
\hhline{---||----}
\multirow{3}{*}{CCA}  & \multirow{3}{*}{N/A} & Clock Bids           & 93.2\% (0.20)          & 17.4\% (0.48)   & 18.0\% (0.44)     & 140 (0.9)      \\ 
                      &                      & Clock Bids Raised    & 93.4\% (0.20)          & 30.0\% (0.80)      & 34.9\% (0.39)   & 141 (0.9)      \\ 
                      &                      & Profit Max           & 94.2\% (0.20)          & 30.0\% (0.77)     & 34.5\% (0.39)   & 141 (0.9)      \\ 
\hhline{---||----}
VCG                   & \makecell[c]{-}      & \makecell[c]{-}      & 100.0\% (0.00)             & 42.2\% (0.31)   & \makecell[c]{-} & 1   (0.0)         \\ 
\hhline{---||----}
Random Allocation     & \makecell[c]{-}      & \makecell[c]{-}      & 34.4\% (0.72)         & \makecell[c]{-} & \makecell[c]{-} & \makecell[c]{-} \\ 
\end{tabular}
\end{tablesizeadjustment}

        \end{adjustbox}
        \caption{Results for MLCA, CCA and VCG in the MRVM domain. We
          use $Q^{\text{max}}=500$ in MLCA and $Q=500$ in the
          ProfitMax heuristic of CCA. Results are averages across 50
          auction instances. Standard errors are in
          parantheses.}
        \label{fig:mechanismComparisonMRVM}
\end{table}

Finally, we turn to MRVM, which is a realistically-sized domain, and
the most complex value model we consider. The results for this domain
are shown in Table~\ref{fig:mechanismComparisonMRVM}. We see that MLCA
(using the Quadratic kernel) achieves an efficiency of 96.4\%, while
the CCA (using the Profit Max heuristic) only achieves an efficiency of
94.2\%. Using a paired t-test, we find that this efficiency difference
is highly statistically significant ($p<1e^{-8}$).

There are multiple reasons for the performance advantage of MLCA over
the CCA in this domain. First, MRVM has $2^{98}$ bundles, and,
importantly, the bidders are interested in a very large subset of
those bundles. This means that $500$ bundle-value pairs is not
sufficient to reasonably describe bidders' preferences.  As a
consequence, the CCA Profit Max heuristic is not nearly as effective
as it was in the simper GSVM and LSVM domains. Second, even though the
Quadratic kernel cannot learn the MRVM model perfectly (unlike in
GSVM), our kernel experiments have shown that the learning error of
the Quadratic kernel in MRVM is also relatively small. Thus, in this large (realistically-sized) domain, the power
of a well-tuned ML algorithm can really shine.

Turning to revenue, we observe that the CCA with the Profit Max
heuristic obtains a much smaller amount of revenue than VCG.  Further, we observe that this low level of revenue is
present not only with the VCG payment rule, but also with a
core-selecting payment rule.  In contrast, MLCA achieves a
significantly higher amount of revenue than the CCA. The most likely
explanation for this is the fact that MLCA explicitly queries the
marginal economies, which leads to good information elicitation in the
marginals and thus higher prices, while the CCA essentially only
focuses on preference elicitation in the main economy.

Finally, we observe that MLCA and the CCA require a comparable number
of rounds, implying that the number of rounds required by MLCA would
not be a limitation in practice.

Overall, we see that MLCA performs exceptionally well in this complex
domain.  We see that when the domain is very large and complex, our
rich ML-based elicitation approach is highly effective, while the CCA
becomes less effective.

\subsection{Manipulation Results}
\label{subsec:ManipulationExperiments}

So far, we have considered truthful behavior on the part of bidders.  But we are also interested in the robustness of the mechanism to non-truthful bidder behavior.  While a full equilibrium analysis of a mechanism as complex as MLCA is beyond the state-of-the-art for either a closed-form or for a computational approach \citep{Bosshard2020JAIR}, we can still provide useful evidence about how the mechanism performs in the face of strategic bidders.  Specifically, we consider the effect of a single bidder attempting a unilateral deviation towards strategic play, holding all other bidders at truthful bidding.   

It remains to specify what strategy our single manipulating bidder, $i$, should employ.  To do so, we reason carefully about MLCA in order to develop a potentially useful strategy.  In this exercise we will allow the bidder access to private information that would not be available to bidders; doing so only strengthens potential power of the strategy and thus the exercise.
For reasons described in Section~\ref{SEC:Incentives}, it will be very difficult for a bidder to gain advantage through manipulating the main economy.  Therefore, we target our strategy at the marginal economy as follows: (a) for those bundles $i$ may win, bid truthfully; (b) for all other bundles overbid by as much as possible without winning the bundle. This approach seeks to drive out competition in $i$'s marginal economy, lowering that economy's social welfare, and thus decreasing $i$'s VCG payment.  

We instantiate the approach by granting bidder $i$ access to the
following information:
\begin{align}
    V_R &= \max_{a \in \mathcal F : a_j \in R_j\forall j \in N} V(a) \label{eq:vr}\\
    V_T(b) &= \max_{a\in\mathcal{F} : a_i = b} V(a) \label{eq:vq}
\end{align}
The scalar $V_R$ \eqref{eq:vr} is the social welfare of the best allocation available among the elicited bundles up to the current round (note that this calculation uses the bids of the other bidders evaluated at their true values).
For a given bundle $b$, $V_T(b)$ \eqref{eq:vq}, is the social welfare of the best allocation where  bidder $i$ is fixed to bundle $b$ and the other bidders are \emph{not} restricted to their reports (note that this calculation uses the full truthful valuation of the other bidders).
Both of these pieces of information contain powerful private information that would not normally be available to bidders outside the confines of the present exercise. 

We then propose the following strategy for bidder $i$, when queried bundle $b$:
\begin{equation}
    \hat{v}(b) = 
    \begin{cases}
       v_i(b) & \text{if } V_R \leq V_T(b) \\
       v_i(b) + z(V_R-V_T(b)) & \text{otherwise}
  \end{cases}
\end{equation}
Where $0 \leq z < 1$ is a parameter specifying the amount of overbidding. 
The strategy works by using $V_T(b)$ as a threshold: so long as the best allocation at reports, $V_R$, is below this value any overbid might change the allocation (assuming other bidders are truthful), and so to be ``safe'' the bidder remains truthful.  Otherwise the bidder can overbid by up to the amount $V_R$ is in excess of the threshold and still be ``safe'' from changing the allocation (assuming other bidders are truthful).  Thus, by construction and as desired, bidder $i$ will never win a misreported bundle when all other bidders report truthfully, as they do in our experiments.

\begin{table}[tb]
\centering
\begin{adjustbox}{max width=\textwidth}
\begin{tabular}{l||ccc|ccc|ccc}
 & \multicolumn{3}{c|}{\bfseries Local} & \multicolumn{3}{c|}{\bfseries Regional} & \multicolumn{3}{c}{\bfseries National} \\
\multicolumn{1}{c||}{Strategy} & \makecell{Social\\Welfare} & \makecell{Marginal Economy\\Social Welfare} &      Utility & \makecell{Social\\Welfare} & \makecell{Marginal Economy\\Social Welfare} &       Utility & \makecell{Social\\Welfare} & \makecell{Marginal Economy\\Social Welfare} &          Utility \\
\hhline{=::===:===:===}
Truthful          &     9863 (128) &                      9863 (128) &  0.02 (0.01) &     9863 (128) &                      9849 (128) &  14.14 (3.00) &     9863 (128) &                      8004 (147) &  1859.21 (95.49) \\
Overbidding (25\%) &     9895 (128) &                      9895 (128) &  0.00 (0.00) &     9871 (128) &                      9854 (127) &  16.39 (4.49) &     9887 (130) &                      8003 (147) &  1883.41 (95.89) \\
Overbidding (50\%) &     9890 (129) &                      9890 (129) &  0.09 (0.06) &     9882 (129) &                      9866 (127) &  16.16 (4.19) &     9843 (125) &                      7981 (146) &  1861.68 (96.17) \\
Overbidding (75\%) &     9892 (127) &                      9892 (127) &  0.06 (0.05) &     9865 (127) &                      9848 (126) &  16.48 (4.30) &     9834 (130) &                      7984 (143) &  1849.95 (89.53) \\
Overbidding (99\%) &     9852 (128) &                      9852 (128) &  0.01 (0.01) &     9872 (126) &                      9857 (126) &  14.77 (4.09) &     9816 (127) &                      7978 (146) &  1837.81 (94.11) \\
\hhline{-||---------}
ANOVA p-value     &          0.999 &                           0.999 &        0.362 &          1.000 &                           1.000 &         0.992 &          0.996 &                           1.000 &            0.998 \\
\end{tabular}

\end{adjustbox}
\caption{Manipulation Experiments for MRVM.  Entries are the average of 50 runs.}
\label{tab:manipulationMRVM}
\end{table}

To evaluate this strategy, we ran experiments for all three domains and with all bidder types separately employing this strategy.  We evalauted $z\in\{0.25,0.5,0.75,0.99\}$ with the quadratic kernel. Results for MRVM are reported in Table~\ref{tab:manipulationMRVM}; we defer the results for GSVM and LSVM to \refapp{app:manipulationExperiments}. 
In the last row of the table we provide the p-value of the one-way ANOVA to test for statistically significant differences across all applied strategies for each column respectively.  We observe that none of the manipulation strategies we consider leads to a statistically significant improvement in the bidder's utility.  Thus, even when providing the bidder with access to powerful information not normally present (i.e., an oracle capable of calculating $V_R$ and $V_T(b)$), the bidder is not able to significantly improve their utility under the proposed strategy targetting the bidder's marginal economy.  This provides further evidence regarding MLCA's robustness against such manipulations.

\section{Conclusion }

In this paper, we have introduced a machine learning-powered iterative
combinatorial auction mechanism we call MLCA. In contrast to prior
designs like the CCA, MLCA does not use prices but \emph{value
  queries} to interact with the bidders. Via simulations, we have
shown that MLCA is able to achieve higher efficiency than the CCA,
even with just a small number of queries per bidder.

Two components in our design are responsible for this efficiency gain:
(1) the ML algorithm learns a bidder's valuation on the whole bundle
space, and (2) in each iteration of the auction, we compute the
(tentatively) optimal allocation based on the learned values, which we
use to decide which value query to ask next to each bidder. To achieve
good incentives, we have drawn on principles from the VCG mechanism to
design \emph{MLCA}.

Our results give rise to promising directions for future research:
First, while we have shown how our elicitation method can be modified
to allow bidders to only report bounds on their values, we have left a
full mechanism based on this bounds-based elicitation method to future
work. Second, while our method of selecting the \emph{initial set of
  queries} uniformly at random from the whole bundle space has worked
surprisingly well, future work could explore more sophisticated active
learning methods for generating this initial set.

\section*{Acknowledgments}
Part of this research was supported by the SNSF (Swiss National
Science Foundation) under grant \#156836 and by the European Research
Council (ERC) under the European Unions Horizon 2020 research and
innovation programme (Grant agreement No. 805542) and by the National
Science Foundation (NSF) under grant no. CMMI-1761163. We thank Fabio
Isler and Manuel Beyeler for excellent research assistance in
implementing some of the algorithms and running some of the
experiments.

\endmaincontent

\beginsupplementpreamble
\externalrefmain %

\RUNTITLE{Supplement for CAs via ML-powered Preference Elicitation}
\TITLE{Supplement for Combinatorial Auctions via \\
Machine Learning-based Preference Elicitation}

\ARTICLEAUTHORS{%
\AUTHOR{Gianluca Brero}
\AFF{University of Zurich, \EMAIL{brero@ifi.uzh.ch}}
\AUTHOR{Benjamin Lubin}
\AFF{Boston University, \EMAIL{blubin@bu.edu}}
\AUTHOR{Sven Seuken}
\AFF{University of Zurich, \EMAIL{seuken@ifi.uzh.ch}}
}
\maketitle
\relax

\tableofcontents
\endsupplementpreamble
\beginsupplementseparator
\newpage
\section*{Appendices}
\endsupplementseparator

\beginsupplementcontent

\begin{APPENDICES}

\renewcommand*{\theHsection}{App.\thesection}

\section{Full Versions of the Query Module and of MLCA}
\label{app:FullVersionsOfAlgorithms}

\begin{algorithm}[H]
        \SetEndCharOfAlgoLine{;}
        \SetKwRepeat{Do}{do}{while}
        \nonl\textbf{function} NextQueries$_\mathcal A$$(I,R, S)$\;
        \nonl \textbf{parameters:} {profile of ML algorithms $\mathcal{A}$}\;
        \nonl \textbf{inputs:} {index set of bidders for the economy to be considered $I$; profile of reports $R$;\\\nonl profile of queries that have already been generated in this auction round $S$}\;
        \textbf{foreach} bidder $i\in I$ : $\tilde v_{i} := \mathcal{A}_i (R_{i})$; \text{\hspace{0.5cm}$\setminus\setminus $\textbf{Learning Step:} learn valuations using ML algorithm}\\
                select $\tilde{a} \in \argmax_{a \in {\mathcal F}} \sum_{i\in I}\tilde v_{i} (a_{i})$; \text{\hspace{0.5cm}$\setminus\setminus $\textbf{Optimization Step} (based on learned valuations)} \label{query module:OptStep}\\
assign new query profile: $q=\tilde{a}$ (i.e., for each $i \in I: q_i = \tilde{a}_i$)\;%
        \ForEach
                {$i \in I$}
                {
        \uIf  {bundle $q_i$ has already been queried or generated before (i.e., $q_i \in R_i \cup S_i)$}                  {
                                        define set of allocations containing a new query for $i$: $\mathcal F' :=  \{a\in \mathcal F : a_i \neq x, \forall x\in         R_i \cup S_i\, \}$\;
                select  $\tilde{a} \in \argmax_{a \in {\mathcal F'}} \sum_{i'\in I}\tilde                 v_{i'} (a_{i'})$; \text{\hspace{0.5cm}$\setminus\setminus $\textbf{Optimization                 Step} (with restrictions) }\\
overwrite new query for bidder $i$: $q_i = \tilde{a}_i$\;
                }
        }
        \textbf{return} profile of new queries $q$\;
        \caption{Machine Learning-powered Query Module (full version, including set $S$)}
        \label{alg:ML_Elicitation_Full}
\end{algorithm}

\newpage

\begin{algorithm}[H]
        \SetEndCharOfAlgoLine{;}
        \nonl \textbf{parameters:} profile of ML algorithms $\mathcal A$; maximum \# of queries per bidder $Q^{\text{max}}$; \\\nonl\# of initial queries $Q^{\text{init}} \leq Q^{\text{max}}$; \# of queries per round $Q^{\text{round}}$; maximum \#\ of push bids per bidder  $P^{\text{max}}$\;
        \textbf{foreach} bidder $i\in N$: receive $P_i \leq P^{max}$ push bids\;
        \textbf{foreach} bidder $i\in N$: ask the bidder to report his value for $Q^{\text{init}}$ randomly chosen bundles\;

Let $R= (R_1,...,R_n)$ denote the initial report profile, where each $R_i$ is $i$'s set of bundle-value reports\;
Let $T=\floor{(Q^{\text{max}} - Q^{\text{init}})/Q^{\text{round}}}$ denote the total number of auction rounds and $t=1$ the current round\;
        \While(\tcp*[f]{ Auction round iterator}){$t \leq T$} {
                Let $S=(S_{1},..., S_{n})$ denote the profile of queries for this auction round, with each $S_i=\emptyset$\;
                \ForEach{bidder $i\in N$}
                {
                        Sample a set of bidders $N'$ from $N\setminus \{i\}$ with $|N'|=Q^{\text{round}}-1$ \;
                        \ForEach
                {$i' \in N'$}
                {
                Generate query profile $q:=$ \emph{NextQueries}$_{\mathcal A_{-i'}}$$(N\setminus\{i'\},R_{-i'},S_{-i'})$\tcp*{Queries for ME}
                For bidder $i$: add $q_i$ to the queries generated for this round, i.e., $S_i = S_i \cup \{q_i\}$\;
                }
        }
Generate query profile $q:=$ \emph{NextQueries}$_\mathcal A$$(N,R,S)$\tcp*{Queries for the main economy}
\textbf{foreach} bidder $i \in N$: add $q_i$ to the queries generated for this round, i.e., $S_i = S_i \cup \{q_i\}$\;
        \textbf{foreach} bidder $i \in N$: send  new queries $S_i$ to bidder $i$ and wait for reports\;%
\textbf{foreach} bidder $i \in N$: receive bundle-value reports $R_i'$ and add them to $R_i$, i.e., $R_i = R_i\cup R_i'$\;
$t=t+1$\;
        }
Let $\hat{v}_i(\cdot)$ denote bidder $i$'s \emph{report function} capturing his bundle-value reports $R_i$, $\forall i \in N$\;
        Compute final allocation: $a^{\text{MLCA}}\in \argmax_{a \in \mathcal{F}_R} \sum_{i \in N} \hat{v}_i (a_i)$\;
        \textbf{foreach} bidder $i \in N$: compute his payment
        \begin{equation}
        \label{eq:payment_MLCA_Full}
        p^{\text{MLCA}}_i = \sum_{j \in N \setminus \{i\}} \hat v_j ( a^{-i}_j) - \sum_{j \in N \setminus \{i\}}\hat
        v_j ( a^{\text{MLCA}}_j), \text{\hspace{1cm} where } a^{-i}\in \argmax_{a \in \mathcal{F}_R} \sum_{j \in N \setminus \{i\}} \hat{v}_j (a_j) ;
        \end{equation}\\
        Output  allocation $a^{\text{MLCA}}$ and payments $p^{\text{MLCA}}$\;
        \caption{Machine Learning-powered Combinatorial Auction (MLCA) (full version) }
        \label{alg:MLCA_Full}
\end{algorithm}

\section{Learning Error and Imputed Approximate Clearing Prices}
\label{app:CE}

Recall that in MLCA, bidders submit bundle-value reports, while prices
are used for elicitation in most prior work on iterative CAs
\citep[e.g.][]{parkes2006mit}.  Typically, the goal of such price-based
elicitation is to obtain approximate clearing prices. Here we relate
our mechanism to the rest of the literature by showing how to obtain a
price-based interpretation of the elicitation performed by MLCA's
query module.  Specifically, we describe how to impute approximate
clearing prices that are implicit in the elicitation.  We provide a
bound on how close these prices are to clearing prices, based upon a
bound on the learning error of the ML algorithm.  We next formalize
these concepts.

To begin, we introduce a very general concept of \textit{prices}
(allowing for non-anonymous bundle prices). We let $\pi =
(\pi_1,..,\pi_n)$ denote the \emph{price profile}, where each $\pi_i$
is bidder $i$'s \textit{price function}, with $\pi_{i}(x)$ denoting
bidder $i$'s price for any given bundle $x\in\mathcal X$.
Next, we define a competitive equilibrium (CE):
\begin{definition}[Competitive equilibrium]
Given prices $\pi$, we define each bidder $i$'s demand set $d^{\pi}_i$
as the set of bundles that maximize his utility at $\pi$:
\begin{equation}\label{eq:demandset}
d^{\pi}_i = \argmax_{x\in\mathcal X} \left(v_i(x) - \pi_i(x)\right)
\end{equation}
Similarly, we can define the seller's supply set $s^{\pi}$ as the set
of allocations that are most profitable at $\pi$:
\begin{equation}\label{eq:supplyset}
s^{\pi} = \argmax_{a\in\mathcal F} \sum_{i\in N} \pi_i(a_i)
\end{equation}
We say that prices $\pi$ and allocation $a$ are in \textbf{competitive
  equilibrium} if $a_i\in d^{\pi}_i \,\,\, \forall i \in N$ and $a\in
s^{\pi}$.  Any prices that are part of a CE, are called
\textbf{clearing prices}.
\end{definition}
\noindent
A special case of the first welfare theorem holds that any competitive
equilibrium allocation is also efficient
\citep[16.C-D]{mas1995microeconomic}.  Prior work on iterative CAs has
exploited this property by iteratively updating prices until a CE is
found \citep[e.g., \textit{i}Bundle,][]{parkes1999bundle}.%
\footnote{We note that the iterative VCGs mechanisms by
  \citet{mishra2007ascending} and \citet{de2007ascending} go beyond CE
  prices and find \emph{universal competitive equilibrium (UCE)
    prices}, which are specific clearing prices that contain all the
  information necessary to compute a VCG outcome.}
This approach is motivated by the fact that any auction that finds an
efficient allocation must reveal CE prices
\citep{nisan2006communication}. However, in the worst case, this may
require an exponential amount of communication (potentially quoting a
different price for each bundle).  Recall that in MLCA, we limit the
amount of information exchanged (via the query cap). Thus, due to the
result by \cite{nisan2006communication}, we cannot guarantee finding a
CE. In fact, MLCA does not even use prices to interact with bidders
(in contrast to the CCA or iterative VCG mechanisms). However, at each
round of the auction, we can \emph{impute} prices based on the learned
valuations in that round. This provides insight into how bundles are
being \emph{implicitly priced} in each round of MLCA.

To make this connection, we will need the following relaxation of
clearing prices:
\begin{definition}[$\delta$-Approximate Clearing Prices]
        Given prices $\pi$, we define each bidder $i$'s $\beta$-demand set
        $d_i^{\pi,\beta}$ as the set of bundles that maximize his utility at
        $\pi$ when subsidized by $\beta$:
        \begin{equation}\label{eq:deltademandset}
                d_i^{\pi,\beta} = \left\{x\in\mathcal X : \left(v_i(x) -
                \pi_i(x)\right) + \beta \geq \left(v_i(x') - \pi_i(x')\right) \forall
                x' \in d_i^\pi \right\}
        \end{equation}
        Similarly, we can define the seller's $\gamma$-supply set
        $s^{\pi,\gamma}$ as the set of allocations that are most profitable at
        $\pi$ when subsidized by $\gamma$:
        \begin{equation}\label{eq:deltasupplyset}
                s^{\pi, \gamma} = \left\{ a\in\mathcal{F} :
                \sum_{i\in N} \pi_i(a_i) + \gamma \geq
                \sum_{i\in N} \pi_i(a_i') \forall a' \in s^\pi \right\}
        \end{equation}
        We say that prices $\pi$ are $\delta$-\textbf{approximate clearing
                prices} if, at any allocation $a^*$, the following holds:
        $$
        a^*_i\in d_i^{\pi,\beta_i} \text{ and }
        a^*\in s^{\pi,\gamma} \text{ and }
        \delta \geq \sum_{i \in N} \beta_i + \gamma.
        $$
\end{definition}
\noindent That is, $\delta$-approximate clearing prices are those
prices that would be clearing, if the participants were subsidized by
an amount not greater than $\delta$.

It turns out that given a bound on the learning error of the ML
algorithm, we can bound the degree of approximation in the price
profile:%
\begin{proposition}
\label{prop:ApproximateClearingSimple}
Let $\tilde v$ be a learned valuation profile, $\pi$ be a profile of
clearing prices for $\tilde v$, and $a^*$ be an efficient allocation
at truthful values $v$.  Assume that the learning errors are bounded
as follows: for each bidder $i$, $\max_{x\in\mathcal X} |\tilde
v_i(x)- v_i(x)|\le \delta_1$ and $|\tilde v_i(\tilde{a}_i)-
v_i(\tilde a_i)|\le \delta_2$.  Then, $\pi$ is a $\big(n( \delta_1+
\delta_2)\big)$-approximate competitive equilibrium price profile for
$v$.
\end{proposition}
\begin{proof}
Follows as a special case of
Proposition~\ref{prop:ApproximateClearing}, below.
\end{proof}
Proposition~\ref{prop:ApproximateClearingSimple} states that clearing
prices for the learned valuation $\tilde{v}$ will be
$\delta$-approximate clearing prices at the true valuations, with a
$\delta$ that depends on the quality of the ML algorithm.  That is,
one would need to inject at most $n( \delta_1+ \delta_2)$ into the
market to induce the bidders and the seller to trade the allocation
$a^*$ at prices $\pi$.  Accordingly if we can move towards CE in the
learned valuations, we are moving towards approximate CE at the true
valuations.

Next, we slightly generalize
Proposition~\ref{prop:ApproximateClearingSimple}, by allowing the
prices $\pi$ to be only approximately clearing:
\begin{proposition}
\label{prop:ApproximateClearing}
Let $\tilde v$ be a learned valuation profile and $a^*$ be an
efficient allocation at truthful values $v$.  Assume that the learning
errors are bounded as follows: for each bidder $i$,
$\max_{x\in\mathcal X} |\tilde v_i(x)- v_i(x)|\le \delta_1$ and
$|\tilde v_i(\tilde{a}_i)- v_i(\tilde a_i)|\le \delta_2$.
Let $\pi$ be a $\delta_3$-approximate competitive equilibrium price
profile for $\tilde v$.
Then, $\pi$ is a $\big(n(\delta_1+\delta_2)+\delta_3\big)$-approximate
competitive equilibrium price profile for $v$.
\end{proposition}
\begin{proof}
Let $\delta_3^s$ be the transfer that should be made to the seller to
trade $\tilde{a}$ at prices $\pi$.  Then, the transfer that needs
to be made to the seller to trade $a^*$ at prices $\pi$ is:
\begin{equation}\label{eq:TransferSeller}
\tau_s= \sum_{i\in N} \pi_i(\tilde a_i) - \pi_i (a^*_i)+\delta_3^s.
\end{equation}
Given $\bar x_i \in d_i^{\pi}$, then the transfer that needs to be
made to each bidder $i$ to trade $a^*$ at prices $\pi$ is
\begin{equation}\label{eq:TransferBidder}
\tau_i= v_i(\bar x_i) - \pi_i(\bar x_i) - \big( v_i(a^*_i) - \pi_i(a^*_i)
\big).
\end{equation}
$\pi$ is a $\delta_3$-approximate competitive equilibrium price
profile for $\tilde v$, so for each bidder $i$ there exist a transfer
$\delta_3^i$ such that
\begin{equation}\label{eq:CEProofInequality}
\tilde v_i(\tilde{a}_i) - \pi_i(\tilde{a}_i) + \delta_3^i \ge
\tilde v_i(\bar x_i) - \pi_i(\bar x_i),
\end{equation}
and
\begin{equation}\label{eq:TotalTransfer}
\sum_{i\in N} \delta_3^i + \delta_3^s\le\delta_3.
\end{equation}
After adding and subtracting $\tilde v_i(\bar x_i)$ to $\tau_i$ in
Equation~\eqref{eq:TransferBidder} and using the
inequality in~\eqref{eq:CEProofInequality}, we obtain
\begin{equation}\label{eq:CEProofInequality2}
\tau_i \le v_i(\bar x_i) - \tilde v_i(\bar x_i) + \tilde
v_i(\tilde{a}_i) - v_i(a^*_i) + \pi_i (a^*_i) - \pi_i(\tilde a_i)+ \delta_3^i.
\end{equation}
Considering the inequalities
in~\eqref{eq:TotalTransfer}~and~\eqref{eq:CEProofInequality2} and
Equation~\eqref{eq:TransferSeller}, we have the following bound on the
overall transfer that should be made to induce $a^*$:
\begin{equation}
\tau_s + \sum_{i\in N} \tau_i \le \sum_{i\in N} v_i(\bar x_i) - \tilde
v_i(\bar x_i) + \sum_{i\in N} \tilde v_i(\tilde{a}_i) -
v_i(a^*_i)+\delta_3.
\end{equation}
Note that $\sum_{i\in N} \tilde v_i(\tilde{a}_i) - v_i(a^*_i) \le
\sum_{i\in N} \tilde v_i(\tilde{a}_i) - v_i(\tilde{a}_i) $, as
$a^*$ is an efficient allocation for $v$. We then have
\begin{equation}
\tau_s + \sum_{i\in N} \tau_i \le \sum_{i\in N} v_i(\bar x_i) - \tilde
v_i(\bar x_i) + \sum_{i\in N} \tilde v_i(\tilde{a}) -
v_i(\tilde{a})+\delta_3 \le n( \delta_1+ \delta_2)+\delta_3,
\end{equation}
which concludes our proof.
\end{proof}
In words, Proposition~\ref{prop:ApproximateClearing} tells us how
close to clearing prices at the true valuations, $v$, is a price
profile $\pi$ that are $\delta_3$-approximate clearing prices at the
learned valuation, $\tilde v$.\footnote{Note that, starting from the
  same approximately-clearing price profile $\pi$, we could now derive
  a bound on the efficiency loss in $\tilde a$. However, it would be
  the same bound we have already proven in the Proposition.}

There are many price vectors we could use to instantiate
Propositions~\ref{prop:ApproximateClearingSimple}
and~\ref{prop:ApproximateClearing}.  One way to instantiate the
propositions is to \emph{impute} as the price profile $\pi =
\tilde{v}$ (i.e. use the learned values as prices).
\footnote{In fact, these will be the seller-optimal imputed prices for
  this $\delta$.  In general there will be a set of price profiles
  that are all $\delta$-approximately clearing that provide more or
  less of the surplus to different players.  Here, without loss we
  adopt the seller-optimal prices to simplify the
  exposition.}%
Such prices are a natural choice because they meet the exact clearing
condition for the buyers and the seller.
We want to emphasize such imputed prices are only implicit when MLCA
is run. Importantly, they inherit the structure of learned value
$\tilde{v}$, which can be very complex, depending on the ML algorithm
used. In particular, the prices will in general be non-anonymous, and
if we use high-dimensional (non-linear) ML algorithms, then they will
also be high-dimension bundle prices. Note that this implicit
high-dimensional price structure enables MLCA to find an approximate
CE (where the approximation depends on the learning error of the ML
algorithm) while approaches based on linear prices may severely
limited.

Together, these prices and
Propositions~\ref{prop:ApproximateClearingSimple}
and~\ref{prop:ApproximateClearing} allow us to interpret the learned
value functions, as effectively, imputed approximate clearing prices.
Introducing this price-based interpretation allows us to explain how
our mechanism is related to the work by \cite{lahaie2004applying} who
were among the first to establish a theoretical connection between
preference elicitation in CAs and machine learning. In this prior
work, they proposed an iterative ML-based elicitation paradigm which
also uses learned valuations to drive the elicitation
process. However, while MLCA is based on value reports, their
algorithm critically requires demand queries (i.e., it communicates
ask prices to bidders in every round). The advantage of their approach
is that it guarantees finding a competitive equilibrium. However, due
to the result by \cite{nisan2006communication}, if their approach was
applied in a general setting, it would require communicating
exponentially-sized prices to the agents in every round, which makes
it impractical in such settings.

In contrast to the approach by \cite{lahaie2004applying}, the CCA is
designed with practical applications in mind (like MLCA). A guiding
principle underlying the clock phase of the CCA is that it aims to
find approximate CE prices. However, as a practical mechanism, it
cannot use exponential communication; instead, it uses linear prices
and a limited number of rounds. Thus, due to
\cite{nisan2006communication}, it cannot guarantee finding a true
CE. Note that the supplementary round of the CCA (allowing bidders to
submit up to 500 bundle-value pairs in one final round) is designed to
address the potential inefficiencies that remain at the end of the clock
phase (e.g., due to using linear prices that may be out of
equilibrium). Note that both, MLCA and CCA are practical auction
designs, but they restrict the amount of information exchanged in
different ways. Moreover, while the CCA explicitly aims to find an
approximate CE by using demand queries, MLCA can be interpreted as
doing so implicitly by using an ML algorithm on value reports.

\section{Winner Determination for Dot-Product and RBF Kernels}
\label{app:WDs}
In this appendix we show how to formulate Problem~\eqref{eq:WDSVRGeneric} under dot-product and RBF kernels as integer programming problems (IPs). That we can do this is perhaps surprising, as it requires us to compactly encode our \emph{non-linear} social welfare objective in a \emph{linear} program, a feat we manage by exploiting the structure of the kernel functions we consider and our binary bundle encoding.

We start by noticing that, under dot-product kernel functions, each kernel evaluation in the objective of Problem~\eqref{eq:WDSVRGeneric} can be formulated as $\kappa(x,x_{ik}) = \bar \kappa(x \cdot x_{ik}) = \bar \kappa(\tau)$, where $\tau\in\{0,...,\bar \tau\}$ is the number of items bundles $x$ and $x_{ik}$ have in common, and $\bar \tau$ is the size of bundle $x_{ik}$. Similarly, under RBF kernels, we have that each kernel evaluation in Problem~\eqref{eq:WDSVRGeneric} can be formulated as $\kappa(x,x_{ik}) = \bar\kappa(||x - x_{ik}||) = \bar\kappa(\tau)$, where $\tau\in\{0,...,\bar \tau\}$ is the number of items contained in only one among bundles $x$ and $x_{ik}$, and $\bar \tau$ is the total number of items $m$. In both cases, we can introduce $\bar \tau+1$ binary variables $z_{ik\tau}$, each indicating the value $\tau$ where the kernel function should be evaluated, and encode each kernel evaluation $\kappa(x,x_{ik})$ as
\begin{align}\label{eq:kernelEncoding}
 &\sum_{\tau=0}^{\bar{\tau}} \bar \kappa(\tau)z_{ik\tau} \\
\text{s.t.} &\quad\sum_{\tau=0}^{\bar{\tau}} z_{ik\tau} = 1. \notag
\end{align}

To use this linearized kernel encoding in Problem~\eqref{eq:WDSVRGeneric}, we should also establish the relationship between the allocation variables $a_{ij}$ and the newly introduced binary variables $z_{ik\tau}$. Under dot-product kernels, this relationship can be encoded by introducing the following constraint for each support vector $x_{ik}$:
\begin{equation}\label{eq:contraintsDotProduct}
        \sum_{j \in x_{ik}} a_{ij} = \sum_{\tau=0}^{|x_{ik}|} (\tau+1) z_{ik\tau}-1.
\end{equation}
The left term in Equation~\eqref{eq:contraintsDotProduct} tracks the number of items that bidder $i$'s allocated bundle $a_i$ and the support vector $x_{ik}$ have in common; the right term enforces that only the $z_{ik\tau}$ corresponding to this number gets activated. Under RBF kernels, the relationship between variables $a_{ij}$ and $z_{ik\tau}$ for each $x_{ik}$ can be encoded as
\begin{equation}\label{eq:contraintsRBF}
        \sum_{j \in x_{ik}} (1-a_{ij}) + \sum_{j \notin x_{ik}} a_{ij} = \sum_{\tau=0}^{m} (\tau+1) z_{ik\tau} - 1.
\end{equation}
Here, the left term tracks the number of items that belong to $x_{ik}$ and not to $a_{ij}$ (first sum) and the ones that belong to $a_{ij}$ and not to $x_{ik}$ (second sum).

After integrating the kernel encoding in~\eqref{eq:kernelEncoding} and the constraints in~\eqref{eq:contraintsDotProduct} in Problem~\eqref{eq:WDSVRGeneric}, we derive that the allocation problem for dot-product kernels can be encoded via the following IP:
\begin{align}
        \max_{a_{ij},z_{ik\tau}}&\quad \sum_{i\in N} \sum_{k\in L} ( \alpha_{ik} -  \alpha_{ik}^*) \sum_{\tau=0}^{\bar{\tau}} \bar \kappa (\tau) z_{ik\tau}
        & \quad \text{Problem~\eqref{eq:WDSVRGeneric}} \notag \\
        \text{s.t.} &\quad \sum_{\tau=0}^{\bar{\tau}} z_{ik\tau} = 1 \quad \forall i \in  N,\, k\in[\ell_i] &\quad \text{Encoding~\eqref{eq:kernelEncoding}} \notag\\
        &\quad \sum_{j \in x_{ik}} a_{ij} = \sum_{\tau=0}^{|x_{ik}|} (\tau+1) z_{ik\tau}-1 \quad \forall i \in  N,\, k\in[\ell_i] &\quad \text{Equation~\eqref{eq:contraintsDotProduct}} \notag\\
        &\quad \sum_{i\in N} a_{ij} \leq 1 \quad \forall j\in M &\quad \text{Problem~\eqref{eq:WDSVRGeneric}} \notag
\end{align}
By replacing the constraints in~\eqref{eq:contraintsDotProduct} with the constraints in~\eqref{eq:contraintsRBF}, we obtain the following formulation of the allocation problem under RBF kernels:
\begin{align}
\max_{a_{ij},z_{ik\tau}}&\quad \sum_{i \in N} \sum_{k\in L} ( \alpha_{ik} -  \alpha_{ik}^*) \sum_{\tau=0}^{\bar{\tau}} \bar \kappa (\tau) z_{ik\tau}
& \quad \text{Problem~\eqref{eq:WDSVRGeneric}} \notag \\
\text{s.t.} &\quad \sum_{\tau=0}^{\bar{\tau}} z_{ik\tau} = 1 \quad \forall i \in  N,\, k\in[\ell_i] &\quad \text{Encoding~\eqref{eq:kernelEncoding}} \notag\\
&\quad  \sum_{j \in x_{ik}} (1-a_{ij}) + \sum_{j \notin x_{ik}} a_{ij} = \sum_{\tau=0}^{m} (\tau+1) z_{ik\tau} - 1 \quad \forall i \in  N,\, k\in[\ell_i] &\quad \text{Equation~\eqref{eq:contraintsRBF}} \notag\\
&\quad \sum_{i\in N} a_{ij} \leq 1 \quad \forall j\in M &\quad \text{Problem~\eqref{eq:WDSVRGeneric}} \notag
\end{align}
Note that the size of both the integer programming problems presented above heavily depends on the number of support vectors. In both problems, each support vector $x_{ik}$ introduces two constraints and
$|x_{ik}|+1$ binary variables under dot-product kernels or $m+1$ binary variables under RBF kernels. As discussed in Section~\ref{sec:MLAlgorithm}, one can reduce the number of support vectors by
using a larger $\varepsilon$ in the SVR training problem~\eqref{eq:WDSVRGeneric}, which can be extremely helpful to maintain our mechanism computationally tractable.

\section{Experiments I: Additional Results}
\label{app:kernelAndInsensitivityComparisonGSVMandMRVM}
\subsection{Results for GSVM domain}
\begin{table}[H]
        \centering

\begin{tablesizeadjustment}
\setlength{\tabcolsep}{1.5pt}%
\begin{tabular}{l | c || r | r | r || r | r | r || r | r | r || r | r | r}
Kernel & $\varepsilon$ & 
\multicolumn{3}{c||}{Efficiency} &
\multicolumn{3}{c||}{Learning Error} & 
\multicolumn{3}{c||}{WD Solve Time} & 
\multicolumn{3}{c}{Optimality Gap} \\
\hhline{~|~||---||---||---||---}
&
& \multicolumn{1}{c|}{50} & 
  \multicolumn{1}{c|}{100} & 
  \multicolumn{1}{c||}{200}
& \multicolumn{1}{c|}{50} & 
  \multicolumn{1}{c|}{100} & 
  \multicolumn{1}{c||}{200}
& \multicolumn{1}{c|}{50} & 
  \multicolumn{1}{c|}{100} & 
  \multicolumn{1}{c||}{200}
& \multicolumn{1}{c|}{50} & 
  \multicolumn{1}{c|}{100} & 
  \multicolumn{1}{c}{200} \\

\hhline{=:=*{4}{::===}}

Exponential &  0 &

          93.3\%    &

          95.8\%   &

          92.7\%   &

              12.04   &

              7.52   &

              4.95   &

          3.36s   &

          48.22s   &

          60.00s   &

          0.00   &

          0.08   &

          1.05   \\

\hhline{-|-*{4}{||---}}

Exponential &  1 &

          91.5\%    &

          96.0\%   &

          98.4\%   &

              15.27   &

              9.64   &

              6.63   &

          1.43s   &

          4.01s   &

          8.86s   &

          0.00   &

          0.00   &

          0.00   \\

\hhline{-|-*{4}{||---}}

Exponential &  2 &

          86.2\%    &

          87.9\%   &

          92.2\%   &

              18.04   &

              12.61   &

              8.46   &

          0.79s   &

          1.83s   &

          2.94s   &

          0.00   &

          0.00   &

          0.00   \\

\hhline{=:=*{4}{::===}}

Gaussian &  0 &

          89.1\%    &

          85.0\%   &

          57.9\%   &

              20.67   &

              16.87   &

              14.26   &

          57.39s   &

          60.00s   &

          60.00s   &

          0.11   &

          1.01   &

          7.41   \\

\hhline{-|-*{4}{||---}}

Gaussian &  2 &

          87.7\%    &

          90.4\%   &

          93.7\%   &

              25.05   &

              20.32   &

              16.96   &

          12.68s   &

          41.17s   &

          53.89s   &

          0.00   &

          0.04   &

          0.10   \\

\hhline{-|-*{4}{||---}}

Gaussian &  4 &

          86.2\%    &

          87.9\%   &

          92.2\%   &

              27.67   &

              23.70   &

              20.26   &

          2.69s   &

          9.61s   &

          18.23s   &

          0.00   &

          0.00   &

          0.00   \\

\end{tabular}
\end{tablesizeadjustment}

        \caption{The effect of varying the insensitivity parameter
                $\varepsilon$ on the predictive performance: both efficiency
                of the predicted optimal allocation, and learning error.
                Results are shown for the two expensive kernels,
                Exponential and Gaussian, in the GSVM domain. Results are averages across 50 instances.}
        \label{fig:insensitivityComparisonGSVM}
\end{table}
\begin{table}[H]
        \centering

\begin{tablesizeadjustment}
\setlength{\tabcolsep}{1.5pt}%
\begin{tabular}{l || r | r | r || r | r | r || r | r | r || r | r | r}
Kernel & 
\multicolumn{3}{c||}{Efficiency} & 
\multicolumn{3}{c||}{Learning Error} & 
\multicolumn{3}{c||}{WD Solve Time} & 
\multicolumn{3}{c}{Optimality Gap} \\
\hhline{~||---||---||---||---}
& \multicolumn{1}{c|}{50} & 
  \multicolumn{1}{c|}{100} & 
  \multicolumn{1}{c||}{200}
& \multicolumn{1}{c|}{50} & 
  \multicolumn{1}{c|}{100} & 
  \multicolumn{1}{c||}{200}
& \multicolumn{1}{c|}{50} & 
  \multicolumn{1}{c|}{100} & 
  \multicolumn{1}{c||}{200}
& \multicolumn{1}{c|}{50} & 
  \multicolumn{1}{c|}{100} & 
  \multicolumn{1}{c}{200} \\
\hhline{=*{4}{::===}}
Linear & 90.3\% & 
         90.3\% & 
         91.1\% &
         
             13.20 &
             12.99 &
             12.72 &
        
         0.03s &
         0.04s & 
         0.06s &
         0.00 & 
         0.00 & 
         0.00 \\
\hhline{-*{4}{||---}}
Quadratic & 88.5\% & 
         96.6\% & 
         100.0\% &
         
             8.54 &
             8.54 &
             0.02 &
        
         0.74s & 
         0.77s & 
         0.47s &
         0.00 & 
         0.00 & 
         0.00 \\
\hhline{-*{4}{||---}}
Exponential &          
          91.5\%   &
         
          96.0\%   &
         
          98.4\%   &

             15.27   &
            
             9.64   &
            
             6.63   &

          1.43s   &
         
          4.01s   &
         
          8.86s   &
          
          0.00   &
          
          0.00   &
         
          0.00   \\
\hhline{-*{4}{||---}}
Gaussian &
         
          87.7\%   &
         
          90.4\%   &
         
          93.7\%   &

             25.05   &
            
             20.32   &
            
             16.96   &

          12.68s   &
         
          41.17s   &
         
          53.89s   &
          
          0.00   &
          
          0.04   &
         
          0.10   \\
\end{tabular}
\end{tablesizeadjustment}

        \caption{Comparison of different kernels in the GSVM domain.  The
                Quadratic kernel obtains the best efficiency under the imposed
                computational constraints. Results are averages across 50 instances.}
        \label{fig:kernelComparisonGSVM}
\end{table}
\subsection{Results for MRVM domain}

\begin{table}[H]
        \centering

\begin{tablesizeadjustment}
\setlength{\tabcolsep}{1.5pt}%
\begin{tabular}{l | c || r | r | r || r | r | r || r | r | r || r | r | r}
Kernel & $\varepsilon$ & 
\multicolumn{3}{c||}{Efficiency} &
\multicolumn{3}{c||}{Learning Error} & 
\multicolumn{3}{c||}{WD Solve Time} & 
\multicolumn{3}{c}{Optimality Gap} \\
\hhline{~|~||---||---||---||---}
&
& \multicolumn{1}{c|}{50} & 
  \multicolumn{1}{c|}{100} & 
  \multicolumn{1}{c||}{200}
& \multicolumn{1}{c|}{50} & 
  \multicolumn{1}{c|}{100} & 
  \multicolumn{1}{c||}{200}
& \multicolumn{1}{c|}{50} & 
  \multicolumn{1}{c|}{100} & 
  \multicolumn{1}{c||}{200}
& \multicolumn{1}{c|}{50} & 
  \multicolumn{1}{c|}{100} & 
  \multicolumn{1}{c}{200} \\

\hhline{=:=*{4}{::===}}

Exponential &  0 &

          80.8\%    &

          77.3\%   &

          14.8\%   &

              7.4e+07   &

              5.7e+07   &

              4.1e+07   &

          60.00s   &

          60.00s   &

          60.00s   &

          0.05   &

          0.41   &

          2.66   \\

\hhline{-|-*{4}{||---}}

Exponential &  2048 &

          79.2\%    &

          78.7\%   &

          74.2\%   &

              1.6e+08   &

              1.5e+08   &

              1.3e+08   &

          42.82s   &

          60.00s   &

          60.00s   &

          0.00   &

          0.03   &

          0.11   \\

\hhline{-|-*{4}{||---}}

Exponential &  4096 &

          76.1\%    &

          76.6\%   &

          75.5\%   &

              2.1e+08   &

              2.0e+08   &

              1.9e+08   &

          15.59s   &

          41.55s   &

          57.51s   &

          0.00   &

          0.00   &

          0.01   \\

\hhline{=:=*{4}{::===}}

Gaussian &  0 &

          \makecell[c]{-}    &

          \makecell[c]{-}   &

          \makecell[c]{-}   &

              \makecell[c]{-}   &

              \makecell[c]{-}   &

              \makecell[c]{-}   &

          \makecell[c]{-}   &

          \makecell[c]{-}   &

          \makecell[c]{-}   &

          \makecell[c]{-}   &

          \makecell[c]{-}   &

          \makecell[c]{-}   \\

\hhline{-|-*{4}{||---}}

Gaussian &  16384 &

          82.7\%    &

          82.7\%   &

          83.8\%   &

              6.6e+08   &

              5.8e+08   &

              5.3e+08   &

          60.00s   &

          60.00s   &

          60.00s   &

          0.06   &

          0.07   &

          0.07   \\

\hhline{-|-*{4}{||---}}

Gaussian &  32768 &

          82.4\%    &

          82.0\%   &

          81.7\%   &

              1.0e+09   &

              9.8e+08   &

              9.5e+08   &

          60.00s   &

          59.26s   &

          60.00s   &

          0.04   &

          0.05   &

          0.05   \\

\end{tabular}
\end{tablesizeadjustment}

        \caption[MRVMkernel]{The effect of varying the insensitivity parameter
                $\varepsilon$ on the predictive performance: both efficiency
                of the predicted optimal allocation, and learning error.
                Results are shown for the two expensive kernels,
                Exponential and Gaussian, in the MRVM domain. We do not report results for the Gaussian kernel with $\varepsilon = 0$, because for this parameterization, the solver fails to find a feasible solution within 60 seconds for a large number of instances.}
        \label{fig:insensitivityComparisonMRVM}
\end{table}

\begin{table}[H]
        \centering

\begin{tablesizeadjustment}
\setlength{\tabcolsep}{1.5pt}%
\begin{tabular}{l || r | r | r || r | r | r || r | r | r || r | r | r}
Kernel & 
\multicolumn{3}{c||}{Efficiency} & 
\multicolumn{3}{c||}{Learning Error} & 
\multicolumn{3}{c||}{WD Solve Time} & 
\multicolumn{3}{c}{Optimality Gap} \\
\hhline{~||---||---||---||---}
& \multicolumn{1}{c|}{50} & 
  \multicolumn{1}{c|}{100} & 
  \multicolumn{1}{c||}{200}
& \multicolumn{1}{c|}{50} & 
  \multicolumn{1}{c|}{100} & 
  \multicolumn{1}{c||}{200}
& \multicolumn{1}{c|}{50} & 
  \multicolumn{1}{c|}{100} & 
  \multicolumn{1}{c||}{200}
& \multicolumn{1}{c|}{50} & 
  \multicolumn{1}{c|}{100} & 
  \multicolumn{1}{c}{200} \\
\hhline{=*{4}{::===}}
Linear & 83.8\% & 
         83.8\% & 
         82.4\% &
         
             9.5e+07 &
             1.1e+08 &
             7.6e+07 &
         
         0.00s &
         0.00s & 
         0.00s &
         0.00 & 
         0.00 & 
         0.00 \\
\hhline{-*{4}{||---}}
Quadratic & 83.7\% & 
         84.8\% & 
         81.8\% &
         
             9.4e+07 &
             9.3e+07 &
             6.1e+07 &
         
         1.55s & 
         1.37s & 
         2.36s &
         0.00 & 
         0.00 & 
         0.00 \\
\hhline{-*{4}{||---}}
Exponential &          
          79.2\%   &
         
          78.7\%   &
         
          74.2\%   &

             1.6e+08   &
            
             1.5e+08   &
            
             1.3e+08   &

          42.82s   &
         
          60.00s   &
         
          60.00s   &
          
          0.00   &
          
          0.03   &
         
          0.11   \\
\hhline{-*{4}{||---}}
Gaussian &
         
          82.7\%   &
         
          82.7\%   &
         
          83.8\%   &

             6.6e+08   &
            
             5.8e+08   &
            
             5.3e+08   &

          60.00s   &
         
          60.00s   &
         
          60.00s   &
          
          0.06   &
          
          0.07   &
         
          0.07   \\
\end{tabular}
\end{tablesizeadjustment}

        \caption[MRVMkernel2]{Comparison of different kernels in the MRVM domain.  The
                Quadratic kernel obtains the best efficiency under the imposed
                computational constraints.\footnotemark}
        \label{fig:kernelComparisonMRVM}
\end{table}

\footnotetext{Note that these are preliminary results subject to change due to ongoing updates to our code base.}

\section{Quadratic Kernel and Global Synergy Value Model}\label{app:GSVMAndQuadKernel}

\begin{proposition}
        Every valuation of the Global Synergy Value Model domain can be formulated in the feature space where the Quadratic kernel is the inner product.
\end{proposition}
\begin{proof}
        To prove this statement, it is sufficient to show that each valuation $v_i$ of the Global Synergy Value Model (GSVM) domain is a 2-wise dependent valuation function in the sense of \citet{conitzer2005combinatorial} and use Proposition~1 of \citet{brero2017probably}.
        According to \citet{conitzer2005combinatorial}, a 2-wise dependent valuation is any valuation $v_i$ that can be expressed as
        \begin{equation}\label{eq:2wiseVal}
                v_i(x)=\sum_{j\in x} \left( w_j + \sum_{j'\in x:j'< j}w_{j,j'} \right)
        \end{equation}
        where $w_j$ and $w_{j,j'}$ are scalar parameters.
        At the same time, as discussed by \citet{goeree2010hierarchical}, any valuation of the GSVM domain $v^{\text{GSVM}}_i$ can be expressed as
        \begin{equation}\label{eq:gsvmVal}
        v^{{\text{GSVM}}}_i(x)=\sum_{j\in \bar x} v^{{\text{GSVM}}}_{ij} \bigg( 1+0.2\, (|\bar x|-1)\bigg),
        \end{equation}
        where $v^{{\text{GSVM}}}_{ij}$ is the value assigned by $v^{{\text{GSVM}}}_i$ to item $j$ and $\bar x$ is a sub-bundle of $x$ containing only the licenses in $x$ that are of interest to bidder $i$ (i.e., all the licenses $j\in x$ such that $v^{{\text{GSVM}}}_{ij}>0$).
        $v^{\text{GSVM}}_i$ can be expressed as a 2-wise dependent valuation by setting $w_j=v^{{\text{GSVM}}}_{ij}$ and $w_{j,j'} = 0.2 \, \big(v^{{\text{GSVM}}}_{ij}+v^{{\text{GSVM}}}_{ij'}$\big) if both $j$ and $j'$ are of interest to bidder $i$, $w_{j,j'} = 0$ otherwise. Indeed, with these parameters, the value contributed to $v_i(x)$ in Equation~\eqref{eq:2wiseVal} by each item $j$ of interest to bidder $i$ will be the sum of
        \begin{itemize}
                \item $v^{{\text{GSVM}}}_{ij}$ (via $w_j$) and
                \item $0.2\, (|\bar x|-1) \, v^{{\text{GSVM}}}_{ij}$ (via all the $w_{j,j'}$ such that $j'\in \bar x$ and $j\neq j'$),
        \end{itemize}
        which corresponds to Equation~\eqref{eq:gsvmVal}.
\end{proof}

\section{Additional Design Features of MLCA}\label{appendix:additional_features}

In this appendix we provide more details on two additional design features of MLCA that will
likely be important in many domains when applying MLCA: (a) enabling the auctioneer to control the number of rounds of the auction, and (b) enabling the auctioneer to
switch out the payment rule (e.g., to charge core-selecting
payments).\medskip

\noindent\textbf{Controlling the number of rounds.} Recall that the
version of MLCA presented in Algorithm~\ref{alg:MLCA} is slightly
abbreviated to improve clarity and readability. In particular, we left
out additional control structure that would be necessary to provide
the auctioneer with explicit control over the number of rounds, which
however might be important in practice. Note that the total number of
auction rounds~$T$ depends on the maximum number of queries
$Q^{\text{max}}$, the number of initial queries $Q^{\text{init}}$, and
the number of agents $n$ in the following way:
$T=\floor{(Q^{\text{max}} - Q^{\text{init}})/n}$
(Line~\ref{alg:MLCA:TotRounds} of Algorithm~\ref{alg:MLCA}).  We
assume that the auctioneer sets $Q^{\text{max}}$ to control the effort
required by the bidders, and $Q^{\text{init}}$ is then optimized
(given $Q^{\text{max}}$) for the chosen ML algorithm, which would
leave the auctioneer without direct control over the number of auction
rounds. This could lead to an undesirable or even impractical number
of rounds in certain domains. Assume the auctioneer sets
$Q^{\text{max}}=500$ and $Q^{\text{init}}=100$. Then, in an auction
with just a few bidders, the resulting number of rounds may be too
large; for example, $n=2$ would result in $200$ rounds. Conversely, in
an auction with a large number of bidders, the resulting number of
rounds may be too small; for example, $n=100$ would result in just $4$
rounds.  Fortunately, it is straightforward to modify MLCA such that
the auctioneer can directly control the number of rounds. In an
auction with a small number of bidders, the auctioneer can simply go
through the steps of generating queries for the main and marginal
economies (Lines~\ref{alg:MLCA:QueryMain}--\ref{alg:MLCA:EndFor})
multiple times before sending all generated queries to the bidders at
once (Line~\ref{alg:MLCA:SendQueries}). By deciding how many times to
go through these steps, the auctioneer now has more control over the
number of rounds.\footnote{Note that this requires that the function
  \emph{nextQuery} is guaranteed to always generate a \emph{new} query
  for every bidder, which we guarantee in our implementation (see
  Footnote~\ref{FN:BookkeepingQueries} for details).} In an auction
with a large number of bidders, the auctioneer can choose to generate
less than $n-1$ queries per bidder per round, to increase the number
of rounds. Specifically, let $Q^{\text{round}}$ denote the number of
queries per round the auctioneer selects. Then the resulting number of
rounds is $T=\floor{(Q^{\text{max}} -
  Q^{\text{init}})/Q^{\text{round}}}$, which provides the auctioneer
with the desired control over the number of rounds. We provide details
for the later approach in the full version of MLCA in
\refapp{app:FullVersionsOfAlgorithms}.\medskip

\noindent\textbf{Alternatives to VCG Payments.}
Recall that our auction computes final payments by applying the VCG
payment rule (see Section~\ref{ssec:CAEfficiency}) based on the
bundle-value pairs reported by the bidders (see
Line~\ref{alg:MLCA:PaymentComputation} of Algorithm~\ref{alg:MLCA}).
It is well known that VCG payments may be outside the \emph{core}.
Informally, this means that payments may be so low that a coalition of
bidders may be willing to pay more in total than what the seller
receives from the current winners.  As bidders in such a coalition
might complain, payments in the core may likely be desirable to
auctioneers in practice.
Accordingly, many contemporary real-world auction formats, such as the
CCA~\citep{ausubel2006clock}, eschew VCG payments in favor of rules
that charge payments in the (revealed) core, such as the VCG-nearest
payment rule~\citep{day2012quadratic}.
In MLCA, it is straightforward to substitute a core-selecting payment
rule (e.g.  VCG-nearest) for VCG in
Line~\ref{alg:MLCA:PaymentComputation} of Algorithm~\ref{alg:MLCA}.
In practice, different auctioneers may have different desires for the
mechanism, and this will affect their choice of payment rule.
Specifically, when targeting good incentives they may want to use VCG
as the payment rule; when aiming to be robust against defecting
coalitions, they may prefer a core-selectring rule such as
VCG-nearest.  While we believe it is important to be able to support
multiple payment rules in order to facilitate an auctioneer tailoring
the mechanism to their particular domain of application, we adopt the
VCG payment rule in Algorithm~\ref{alg:MLCA} and in the following
section on theoretical properties of MLCA, as it is simpler to
analyze.  In our experiments (Section~\ref{sec:MLCAvsCCAExp}), we
additionally provide the revenue numbers that result from using
VCG-nearest.    %

\section{Manipulation Experiments}
\label{app:manipulationExperiments}

\begin{table}[H]
\centering
\begin{adjustbox}{max width=\textwidth}
\begin{tabular}{l||ccc|ccc}
 & \multicolumn{3}{c|}{\bfseries Local} & \multicolumn{3}{c}{\bfseries National} \\
\multicolumn{1}{c||}{Strategy} & \makecell{Social\\Welfare} & \makecell{Marginal Economy\\Social Welfare} &      Utility & \makecell{Social\\Welfare} & \makecell{Marginal Economy\\Social Welfare} &       Utility \\
\hhline{=::===:===}
Truthful          &    437.5 (3.6) &                     418.1 (3.5) &  19.40 (1.84) &    437.5 (3.6) &                     433.5 (4.0) &  4.01 (1.11) \\
Overbidding (25\%) &    437.5 (3.6) &                     418.1 (3.5) &  19.40 (1.84) &    437.5 (3.6) &                     433.5 (4.0) &  4.01 (1.11) \\
Overbidding (50\%) &    437.5 (3.6) &                     418.1 (3.5) &  19.40 (1.84) &    437.5 (3.6) &                     433.5 (4.0) &  4.01 (1.11) \\
Overbidding (75\%) &    437.5 (3.6) &                     418.1 (3.5) &  19.40 (1.84) &    437.5 (3.6) &                     433.5 (4.0) &  3.96 (1.11) \\
Overbidding (99\%) &    437.4 (3.6) &                     418.1 (3.5) &  19.31 (1.83) &    437.5 (3.6) &                     433.5 (4.0) &  3.94 (1.11) \\
\hhline{-||------}
ANOVA p-value     &          1.000 &                           1.000 &         1.000 &          1.000 &                           1.000 &        1.000 \\
\end{tabular}

\end{adjustbox}
\caption{Manipulation Experiments for GSVM.  Entries are the average of 100 runs.}
\label{tab:manipulationGSVM}
\end{table}

\begin{table}[H]
\centering
\begin{adjustbox}{max width=\textwidth}
\begin{tabular}{l||ccc|ccc}
 & \multicolumn{3}{c|}{\bfseries Local} & \multicolumn{3}{c}{\bfseries National} \\
\multicolumn{1}{c||}{Strategy} & \makecell{Social\\Welfare} & \makecell{Marginal Economy\\Social Welfare} &      Utility & \makecell{Social\\Welfare} & \makecell{Marginal Economy\\Social Welfare} &       Utility \\
\hhline{=::===:===}
Truthful          &    532.6 (4.5) &                     521.0 (5.8) &  11.54 (2.78) &    532.6 (4.5) &                     512.7 (4.2) &  19.88 (2.82) \\
Overbidding (25\%) &    532.3 (4.5) &                     520.6 (5.9) &  11.70 (2.86) &    532.7 (4.5) &                     513.0 (4.2) &  19.64 (2.76) \\
Overbidding (50\%) &    532.6 (4.5) &                     520.7 (5.8) &  11.97 (2.68) &    532.2 (4.5) &                     511.8 (4.4) &  20.38 (2.79) \\
Overbidding (75\%) &    532.4 (4.5) &                     521.1 (5.8) &  11.32 (2.51) &    531.7 (4.5) &                     513.0 (4.2) &  18.70 (2.77) \\
Overbidding (99\%) &    527.8 (4.9) &                     519.8 (5.8) &   7.91 (1.94) &    525.6 (4.6) &                     511.7 (4.2) &  13.97 (2.49) \\
\hhline{-||------}
ANOVA p-value     &          0.933 &                           1.000 &         0.792 &          0.784 &                           0.999 &         0.461 \\
\end{tabular}

\end{adjustbox}
\caption{Manipulation Experiments for LSVM.  Entries are the average of 99 runs (One run did not terminate for numerical reasons within CPLEX).}
\label{tab:manipulationLSVM}
\end{table}

\end{APPENDICES}

\endsupplementcontent

\bibliographystyle{aeaSS}
\bibliography{Bibliography}

\end{document}